\documentclass[11pt,reqno]{amsart}
\pdfoutput=1
\usepackage[margin=1.25in]{geometry}

\usepackage{enumerate}
\usepackage{dsfont}
\usepackage{graphicx,epstopdf}
\usepackage{amsmath, amsthm, amscd, amsfonts, amssymb}
\usepackage{algorithm}
\usepackage{bm}
\usepackage{subcaption}
\usepackage[title]{appendix}
\usepackage{multirow}
\usepackage{verbatim}
\usepackage{braket}
\usepackage{xcolor}
\usepackage{appendix}
\usepackage{mathtools}  

\definecolor{mulberry}{rgb}{0.77, 0.29, 0.55}
\definecolor{lightpink}{rgb}{1.0, 0.71, 0.76}
\definecolor{yellow}{RGB}{240,228,66}
\definecolor{lightblue}{RGB}{0,58,120}
\definecolor{darkblue}{rgb}{0.1,0.1,0.7}
\definecolor{darkred}{rgb}{0.5,0.1,0.1}
\definecolor{darkgreen}{rgb}{0.0,0.42,0.06}
\definecolor{shadecolor}{rgb}{0.85,0.85,0.85}
\definecolor{navyblue}{rgb}{0.0, 0.0, 0.5}
\definecolor{cerisepink}{rgb}{0.93, 0.23, 0.51}
\definecolor{oceanboatblue}{rgb}{0.0, 0.47, 0.75}
\definecolor{red-brown}{rgb}{0.65, 0.16, 0.16}
\definecolor{carnelian}{rgb}{0.7, 0.11, 0.11}

\usepackage[
            colorlinks=true,
            citecolor=blue,
            linkcolor=red,
            filecolor=magenta,      
            urlcolor=cerisepink,
            filecolor=blue,
            breaklinks=true,
            bookmarks=true]{hyperref}
\usepackage[capitalize]{cleveref}

\usepackage{tikz}
\usetikzlibrary{shapes.geometric}
\usetikzlibrary{arrows.meta}
\usetikzlibrary{positioning}
\usetikzlibrary{quantikz2}

\usepackage{pgfplots}
\pgfplotsset{compat=1.8}

\newcommand{\R}{{\mathbb{R}}}

\newtheorem{thm}{Theorem}[section]
\newtheorem{prop}[thm]{Proposition}
\newtheorem{lem}[thm]{Lemma}
\newtheorem{cor}[thm]{Corollary}
\newtheorem{remark}[thm]{Remark}

\newtheorem{result}[thm]{Result}

\theoremstyle{definition}

\usepackage[style=numeric,url=false,sorting=none,backref=true,maxbibnames=99,doi=false]{biblatex}
\DefineBibliographyStrings{english}{%
  backrefpage = {page},
  backrefpages = {pages},
}
\begin{document}



\title[Quantum Circuits for Korobov Functions]{Approximating Korobov Functions via Quantum Circuits} 
\author[J. Aftab]{Junaid Aftab} \address{Department of Mathematics, University of Maryland, College Park, 4176 Campus Drive, College Park, MD 20742, United States} \email{junaid@umd.edu} 

\author[H. Yang]{Haizhao Yang} \address{Department of Computer Science, University of Maryland, College Park, 4176 Campus Drive, College Park, MD 20742, United States} \email{hzyang@umd.edu}
\date{}


	\begin{abstract}
    Approximation theory provides a useful framework for assessing the expressive capacity and computational complexity of quantum circuits. In this work, we design quantum circuits that approximate $d$-dimensional functions in the Korobov function space. Our construction combines quantum signal processing (QSP) with the linear combination of unitaries (LCU) method to implement circuits that generate Chebyshev polynomials. We analyze the resulting approximation error rates and characterize the computational complexity of the circuits. Because the Korobov function space is a subspace of certain Sobolev spaces, our results establish a theoretical foundation for approximating a broad class of functions on quantum computers.
	\end{abstract}
\maketitle

\vspace{-2.0em}
	\section{Introduction}\label{intro}
	Quantum computing shows significant potential for addressing high-dimensional problems, such as partial differential equations (PDEs) \cite{childs2021high, liu2021efficient, jin2023quantum, sato2024hamiltonian}. Most existing quantum algorithms, however, are designed for fault-tolerant quantum computers, which current hardware cannot yet support. This limitation has motivated growing interest in algorithms tailored to near-term quantum devices. Among these approaches, parameterized quantum circuits (PQCs) have emerged as a promising framework \cite{Benedetti2019PQC-machinelearning}. Numerous PQC-based algorithms have since been proposed for diverse applications \cite{cerezo2021variational, tilly2022variational, alghassi2022variational, liu2021variational}.
	
	A parameterized quantum circuit can be viewed as the quantum analog of a classical feedforward neural network. In the classical setting, substantial progress has been made toward developing a rigorous theoretical understanding of neural networks. Neural network approximation theory \cite{devore2021neural} provides a systematic framework for analyzing their expressive power, with particular emphasis on deriving approximation error bounds in terms of network width, depth, and the number of neurons \cite{yang2021nnapproxsmooth, yang2023nearly, yang2022optimalapproxrelu, yang2023nearlysobolev}. This research aims to characterize the fundamental limitations of neural networks in approximating functions from different function spaces, independent of the specific learning algorithm or availability of training data.
	
	Similarly, a rigorous theoretical framework for parameterized quantum circuits is essential for assessing quantum machine learning models and their potential to advance classical machine learning tasks. Recent work has established universal approximation theorems for parameterized quantum circuits \cite{2021adrianonequbituniversalapproximant, 2021gotouniversalapproximationqml, 2022zhanpowerlimitationsqnn}, ensuring the existence of circuits capable of approximating target functions to arbitrary accuracy. Yet, the circuit complexity required to realize such approximations on quantum hardware remains largely unexplored. Addressing this gap is critical for strengthening universal approximation results by providing explicit circuit constructions and deriving bounds on their computational complexity.
	
	In this work, we extend the framework introduced by \cite{yu2024provable} to advance the emerging field of \emph{quantum neural network approximation theory}. Specifically, we investigate the use of unparameterized quantum circuits to approximate functions via Chebyshev expansions, focusing on the approximation of \( d \)-dimensional Korobov functions. We derive worst-case complexity bounds on circuit width and depth required to achieve a prescribed level of accuracy.

    The Korobov function space is well-suited for high-dimensional problems because of its compatibility with sparse grid decomposition \cite{bungartz2004sparse}, which alleviates the curse of dimensionality inherent in grid-based numerical methods \cite{bellman1959mathematical}. Sparse grids were originally introduced for numerical integration \cite{korobov1959, smolyak1960interpolation} and have since been applied to solving partial differential equations (PDEs) \cite{zenger1991sparse, balder1996solution, griebel2007sparse, rong2017nodal, chernov2019sparse} as well as to the design of neural network architectures (see \cref{comparison}).

	
	\subsection{Summary of Results}\label{sketch}  
	The Korobov function space on \( [0,1]^d \subseteq \mathbb{R}^d \) is denoted by \( X^{r,p}([0,1]^d) \), and will be formally defined later. Our first main result establishes the existence of a quantum circuit capable of approximating functions in \( X^{2,p}([0,1]^d) \):

	\begin{result}\label{result:informal-result1}
		(Informal version of \cref{prop:approximation-error})
		Let $\epsilon \in (0,1)$ and $p \in \{2,\infty\}$. For $f \in X^{2, p}([0,1]^d)$, there exists a function, $g \in X^{2, p}([0,1]^d)$, defined on $[0,1]^d$ such that
		\begin{equation*}
			\| f- g \|_{L^p([0,1]^d)} \leq \epsilon.
		\end{equation*}
		Moreover, there exists a quantum circuit, $U_{f,\varepsilon}$, that outputs $g$ such that $U_{f,\varepsilon}$ has depth at most 
		$\mathcal{O}  ( d \epsilon^{-(\frac{1}{2}+\frac{1}{d})}  (2 \log^{3/2}_2 (\frac{1}{\epsilon}) )^d  )$  and the width at most $\mathcal{O}(2d +  \epsilon^{- \frac{1}{d}} \log^{\frac{3}{2}}_2 (\frac{1}{\epsilon}) )$. 
	\end{result}
	\begin{remark}
		The estimates given above are looser than those presented in \cref{prop:approximation-error}. These simplified bounds are obtained from the more precise estimates by applying the inequality \( W(x) \leq x \) for \( x \geq 0 \), where \( W \) denotes Lambert's \( W \) function.
	\end{remark}
	The result builds on the observation that functions in \( X^{2,p}([0,1]) \) (for \( p \in \{2,\infty\} \)) can be approximated by linear combinations of products of Chebyshev polynomials of degree 0 and 1. We implement these products over \( [0,1]^d \) using the quantum signal processing (QSP) algorithm and combine them via the linear combination of unitaries (LCU) method. This construction allows us to derive upper bounds on the circuit width and depth in terms of \( \epsilon \) and \( d \). We then extend the result to the case \( 2 < p < \infty \).
	\begin{result}\label{result:informal-result2}
		(Informal version of \cref{cor:compplexity-estimate-2})
		Let $\epsilon \in (0,1)$ and $2 \leq p \leq \infty$. For $f \in X^{2, p}([0,1]^d)$, there exists a function, $g \in X^{2, p}([0,1]^d)$, defined on $[0,1]^d$ such that
		\begin{equation*}
			\| f- g \|_{L^p([0,1]^d)} \leq \epsilon.
		\end{equation*}
		Moreover, there exists a quantum circuit, $U_{f,\varepsilon}$, that outputs $g$ such that $U_{f,\varepsilon}$ has depth is at most 
		\begin{equation}
			\mathcal{O} \left(  
			d (12\beta \log_2 \beta)^{\alpha+\beta} 
			\alpha^{\alpha+\beta} 
			\epsilon^{-\frac{p}{2p-1}\left(1 + \frac{1}{d} \right)}
			\log_2^{\alpha+\beta} \left(\frac{1}{\epsilon}\right)
			\epsilon^{-\frac{p}{d(2p-1)}}
			\right)
		\end{equation}
		and the width is at most 
		\begin{equation}
			\mathcal{O} \left( 2d +  (6\beta \log_2 \beta)^{\alpha}  \alpha^{\alpha} \epsilon^{-\frac{p}{d(2p-1)}} \log_2^{\alpha} \left(\frac{1}{\epsilon}\right)\right). 
		\end{equation}
		Here we have defined $\alpha=(3p-1)/(2p-1)$ and $\beta = \alpha (d-1)$.
	\end{result}

    We summarize the main contributions of our work as follows:
\begin{enumerate}
    \item[(i)] We extend the investigation in \cite{yu2024provable} to the Korobov function space which is a subspace of certain Sobolev spaces. This extension allows exploration of a new function space, thereby broadening the scope of function approximation on a quantum computer.
    \item[(ii)] We construct \emph{unparameterized} quantum circuits. This approach leverages the observation that Chebyshev polynomials can be implemented using a predetermined set of parameters within the quantum signal processing algorithm. Consequently, approximating functions via Chebyshev polynomials emerges as a natural strategy on a quantum computer.
    \item[(iii)] In contrast to classical neural network architectures, which typically employ greater width with limited depth, our quantum circuit constructions feature reduced width but increased depth. This contrast highlights a complementary relationship between the two approaches and suggests a trade-off in architectural design for function approximation.
\end{enumerate}

	\subsection{Related Works}\label{comparison}
	In this section, we provide a brief overview of related work in both the neural network and quantum computing literature.

	\subsubsection{Neural Network Literature}
	Neural network approximation theory for Korobov spaces has advanced considerably, particularly in addressing the curse of dimensionality. Recent work has shown that deep neural networks, especially those with ReLU activations, can efficiently approximate functions in these spaces. For example,   in \cite{montanelli2019new} established error bounds for deep ReLU networks by connecting their approximation capabilities to sparse grid methods. Moreover, Yang and Lu in  \cite{yang2023optimalkorobov} derived nearly optimal approximation rates for deep networks on Korobov functions, achieving \emph{super-convergence} rates that exceed those of traditional approximators. Other recent work \cite{mao2022approximation, li2025higherorderapproximationrates, fang2025korobovCNN}  have explored the use of deep convolutional neural networks (CNNs) for Korobov function approximation, demonstrating improved rates and enhanced expressivity. These results underscore the effectiveness of deep neural networks in learning functions from Korobov spaces and suggest their relevance for quantum circuit design and other high-dimensional applications.

	\subsubsection{Quantum Computing Literature} Universal approximation theorems for parameterized quantum circuits (PQCs) have been the focus of recent studies \cite{2021adrianonequbituniversalapproximant, 2021gotouniversalapproximationqml, gonon2023universalqnnreservoir, manzano2025approximation}. In particular, Manzano et al. \cite{manzano2025approximation} demonstrated the existence of PQCs capable of approximating functions in \(L^p(\mathbb T^d)\), the space of \(2\pi\)-periodic Lebesgue integrable functions, and in \(W^{2,k}(\mathbb T^d)\), the space of \(2\pi\)-periodic Sobolev functions. However, these works do not provide explicit quantum circuit constructions or algorithms. Our work extends presents explicit quantum circuit designs along with corresponding complexity estimates for approximating Korobov functions, which are subspaces of certain Sobolev spaces.

Recent advances have initiated the study of approximation theory for PQCs. Notably, Chen et al. in \cite{yu2024provable} establish non-asymptotic error bounds for PQCs, showing that data re-uploading architectures can efficiently approximate multivariate polynomials and smooth functions. Their results further indicate that, under appropriate smoothness conditions, PQCs can achieve approximation efficiency comparable to or exceeding that of classical deep ReLU networks, particularly regarding circuit depth and parameter count.

While PQC approximation theory does not directly address optimization and training dynamics, these aspects have been studied elsewhere. For example, Khadijeh et al. in \cite{liu2023analytic} analyze the dynamics of wide quantum neural networks (QNNs), providing an analytic framework for convergence behavior during training, and demonstrating exponential decay of residual training error as a function of system parameters. In contrast, You and Wu in \cite{you2021exponentially} examine QNN loss landscapes, showing that an exponentially large number of local minima can pose significant optimization challenges. Additionally, You et al. in \cite{you2022convergencetheoryoverparameterizedvariational} develop a convergence theory for over-parameterized variational quantum eigensolvers, identifying thresholds for parameter counts necessary for efficient convergence, depending on the system and Hamiltonian characteristics.

Finally, \textcite{an2023theoryquantumdifferentialequation} establishes theoretical bounds on quantum algorithms for solving linear ordinary differential equations and introduces fast-forwarding techniques that improve efficiency in specific instances. These contributions collectively strengthen the theoretical foundation of our work by clarifying the capabilities and limitations of PQCs in function approximation and differential equation solving within quantum machine learning frameworks.

	\subsection{Organization}
The remainder of the paper is organized as follows. \cref{section:prelims} introduces the notation and provides background on function theory, quantum computing, and relevant quantum algorithms. \cref{main-results} presents our main results, beginning with the quantum implementation of linear combinations of Chebyshev polynomials, followed by the definition of Korobov function spaces and corresponding approximation results. Finally, \cref{future} concludes the paper with remarks and directions for future research.

\subsection{Acknowledgments} 
We thank Christopher Schwab from ETH Z\"{u}rich for valuable discussions and insightful comments. Junaid Aftab acknowledges the support by the National Science Foundation under the grant DMS-2231533. Haizhao Yang was partially supported by the US National Science Foundation under awards DMS-2244988, DMS-2206333, the Office of Naval Research Award N00014-23-1-2007, and the DARPA D24AP00325-00.

	\section{Preliminaries}\label{section:prelims}
We present preliminary material necessary for the subsequent developments. \cref{subsection:notation} defines the notation used throughout the manuscript. \cref{function-theory} reviews key aspects of function theory, while \cref{quantum-prelims} covers the mathematical foundations of quantum computing. Finally, \cref{quant-algs} provides an overview of the quantum signal processing and linear combination of unitaries algorithms.

	\subsection{Notation}\label{subsection:notation}
We adopt the conventions $0! := 1$ and $0^0 := 1$. The symbols $\mathbb{R}$, $\mathbb{C}$, and $\mathbb{N}$ denote the sets of real numbers, complex numbers, and natural numbers, respectively. For $z \in \mathbb{C}$, we denote its complex conjugate by $z^*$.  For a subset $A \subseteq \mathbb{R},\mathbb{C}$, let $\mathds{1}_A(z)$ denote the indicator function:
\begin{equation}
    \mathds{1}_A(z) = 
    \begin{cases}
        1, & \text{if } z \in A, \\
        0, & \text{if } z \notin A.
    \end{cases}
\end{equation}
For $d \geq 1$, a multi-index $\boldsymbol{\alpha}$ is an element of $\mathbb{N}^d$. Its $1$-norm and $\infty$-norm are defined as
\begin{equation}
    \|\boldsymbol{\alpha}\|_1 = \sum_{j=1}^d |\alpha_j|, \quad 
    \|\boldsymbol{\alpha}\|_\infty = \max_{1 \leq j \leq d} |\alpha_j|.
\end{equation}
We define the special multi-indices $\boldsymbol{1} = (1, \ldots, 1)$ and $\boldsymbol{2} = (2, \ldots, 2)$. The following component-wise operations on multi-indices are considered:
\begin{align}
    \boldsymbol{\alpha} \cdot \boldsymbol{\beta} & := (\alpha_1 \beta_1, \ldots, \alpha_d \beta_d), \\
    b \cdot \boldsymbol{\alpha} & := (b \alpha_1, \ldots, b \alpha_d), \quad b \in \mathbb{R}, \\
    c^{\boldsymbol{\alpha}} & := (c^{\alpha_1}, \ldots, c^{\alpha_d}), \quad c \in \mathbb{R}.
\end{align}

For non-negative functions $f,g : \mathbb{N} \to \mathbb{R}^+$, we use the standard complexity notation:
\begin{itemize}
    \item[(i)] $f(n) = \mathcal{O}(g(n))$ if there exist constants $C>0$ and $N \in \mathbb{N}$ such that $f(n) \leq C g(n)$ for all $n \geq N$.
    \item[(ii)] $f(n) = \Omega(g(n))$ if there exist constants $C>0$ and $N \in \mathbb{N}$ such that $f(n) \geq C g(n)$ for all $n \geq N$.
    \item[(iii)] $f(n) = \Theta(g(n))$ if and only if $f(n) = \mathcal{O}(g(n))$ and $f(n) = \Omega(g(n))$.
\end{itemize}

	\subsection{Function Theory}\label{function-theory}
We review key elements of function theory relevant to our work. \cref{orthogonal-polys} discusses orthogonal polynomials, while \cref{function-spaces} introduces the function spaces used in our analysis.

	\subsubsection{Orthogonal Polynomials}\label{orthogonal-polys} 
	Let \( \mathbb{R}[x] \) and \( \mathbb{C}[x] \) denote the polynomial rings in one variable over \( \mathbb{R} \) and \( \mathbb{C} \), respectively, where \( x \) is an indeterminate. For \( p(x) \in \mathbb{C}[x] \), \( p^*(x) \) represents the polynomial in \( \mathbb{C}[x] \) obtained by taking the complex conjugate of each coefficient of \( p(x) \). We will consider the following three examples of orthogonal polynomials in this work:
	
	\begin{enumerate}
		\item[(i)] For $r \in \mathbb N\cup \{0\}, T_r(x) \in \mathbb R[x]$ denotes the degree-$r$ Chebyshev polynomial of the first-kind defined recursively by $T_0(x)=1, T_1(x)= x$ and 
		\begin{align}
			T_r(x) & = 2xT_{r-1}(x) - T_{r-2}(x), \quad r \geq 2
		\end{align}
		on $[-1,1]$. 
		\item[(ii)] Similarly, $U_r(x) \in \mathbb R[x]$ denotes the degree-$r$ Chebyshev polynomial of the second-kind defined recursively by $T_0(x)=1, T_1(x)= 2x$ and 
		\begin{align}
			S_r(x) & = 2xS_{r-1}(x) - S_{r-1}(x), \quad r \geq 2
		\end{align}
		on $[-1,1]$. 
	\end{enumerate}

	\subsubsection{Function Spaces}\label{function-spaces}
Let $[0,1]^d \subseteq \mathbb{R}^d$ denote the unit cube equipped with the $d$-dimensional Lebesgue measure. We consider the following function spaces:
	
	\begin{itemize}
		\item[(i)] Let \( X \) be a set equipped with the counting measure. For $1 \leq p \leq  \infty$, we consider the function space of summable sequences defined as follows:
		\begin{equation}
			\ell^p(X) = \{\boldsymbol \alpha : X \to \mathbb R \,|\, f \text{ is Lebesgue summable and } \|\boldsymbol \alpha \|_{\ell^p(X)} < \infty\}.
		\end{equation}
		The corresponding norm is given by:
		\begin{align}
			\|\alpha\|_{\ell^p(X)} & := 
			\begin{cases} 
				\left( \sum_{x \in X} |\alpha_x|^p \right)^{1/p}, & \text{if } p < \infty, \\[8pt]
				\sup_{x \in X} |\alpha_x|, & \text{if } p = \infty.
			\end{cases}
		\end{align}
		\item[(ii)] For $1 \leq p \leq  \infty$, we consider the function space of Lebesgue integrable functions on $[0,1]^d$ defined as follows:
		\begin{equation}
			L^p([0,1]^d) = \{f : [0,1]^d \rightarrow \mathbb{R} \,|\, f \text{ is Lebesgue measurable and } \|f\|_{L^p([0,1]^d)} < \infty\}.
		\end{equation}
		The corresponding norm is given by:
		\begin{align}
			\|f\|_{L^p([0,1]^d)} & := 
			\begin{cases} 
				\left( \int_{[0,1]^d} |f(\boldsymbol{x})|^p \, d \boldsymbol{x} \right)^{1/p}, & \text{if } p < \infty, \\[8pt]
				\text{ess sup}_{\boldsymbol{x} \in [0,1]^d} |f(\boldsymbol{x})|, & \text{if } p = \infty.
			\end{cases}
		\end{align}
		\item[(iii)] We also consider the function space of Sobolev functions on $[0,1]^d$. Let \( k \in \mathbb{N} \) and \( 1 \leq p \leq \infty \). For each multi-index \( \boldsymbol{\alpha} \in \mathbb{N}^k \) of length \( k \), let $D^{\boldsymbol \alpha}f = D^{\alpha_1} \cdots D^{\alpha_k} f$ denote the $\boldsymbol{\alpha}$-th weak derivative of $f$. The Sobolev space is defined as:     \begin{align}
			W^{k,p}([0,1]^d) := \{ & f \in L^p([0,1]^d) \mid D^{\boldsymbol \alpha}f \in L^p([0,1]^d) \text{ for all } \alpha \in \mathbb{N}^d \text{ with } \| \boldsymbol \alpha \|_1 \leq k \}.
		\end{align}
		The corresponding norm is given by:
		\begin{align}
			\|f\|_{W^{k,p}([0,1]^d)} & := 
			\begin{cases} 
				\left( \sum_{\|\boldsymbol \alpha \|_1 \leq k} \|D^{\boldsymbol \alpha} f\|_{L^p([0,1]^d)}^{p} \right)^{1/p}, & \text{if } p < \infty, \\[8pt]
				\max_{\|\boldsymbol \alpha \|_1 \leq k} \|D^{\boldsymbol \alpha} f\|_{L^\infty([0,1]^d)}, & \text{if } p = \infty.
			\end{cases}
		\end{align}
		When \( p = 2 \), we denote \( W^{k,2}([0,1]^d) \) as \( H^k([0,1]^d) \) for \( k \in \mathbb{N} \). Moreover, 
		\begin{equation}
			H_0^k([0,1]^d) = \{ f \in H^k([0,1]^d) \mid D^{\boldsymbol \alpha}f |_{\partial [0,1]^d} = 0 \text{ for all } \alpha \in \mathbb{N}^d \text{ with } \| \boldsymbol \alpha \|_1 \leq k-1 \},
		\end{equation}
		is the subspace of functions whose weak derivatives of all orders up to 
		$k-1$ vanish on $\partial [0,1]^d$, the boundary of $[0,1]^d$.
	\end{itemize}

	\subsection{Quantum Computing}\label{quantum-prelims} 
This section provides a brief review of the fundamental elements of quantum computing from a mathematical perspective. For a more comprehensive introduction, the reader is referred to \cite{nielsen2002quantum, scherer2019mathematics}.

	\subsubsection{Dirac's Notation}
A quantum system is modeled by a finite-dimensional complex inner product space, known as a Hilbert space. For \( n \geq 2 \), the state of an \( n \)-dimensional quantum system is represented by a vector \( v \in \mathbb{C}^n \) with \(\|v\| = 1\). A quantum state \( v \) is denoted in Dirac notation as \(|v\rangle\), called a \emph{ket}, and its complex conjugate transpose is written as \(\langle v|\), called a \emph{bra}. This convention is known as Dirac's bra-ket notation. Given a basis \(\{\ket{i}\}_{i=1}^n\) for \(\mathbb{C}^n\), any quantum state can be expressed as
\begin{equation}
\ket{v} = \sum_{i=1}^n v_i \ket{i}.
\end{equation}
	
	\subsubsection{Single Qubit Systems} 
	The simplest example of a quantum system is a two-dimensional system, whose state is called a \emph{qubit} or quantum bit. A qubit is a two-dimensional complex vector with unit norm, i.e., \( \ket{v} = (v_1, v_2)^\top \in \mathbb{C}^2 \) with \(\|v\| = 1\). In Dirac notation, a basis for \(\mathbb{C}^2\) is given by
\begin{equation}
    \ket{0} = \begin{pmatrix} 1 \\ 0 \end{pmatrix}, \quad 
    \ket{1} = \begin{pmatrix} 0 \\ 1 \end{pmatrix},
\end{equation}
known as the \emph{computational basis}. An arbitrary qubit in this basis can be written as
\begin{equation}
    \ket{v} = v_1 \ket{0} + v_2 \ket{1}, \quad |v_1|^2 + |v_2|^2 = 1.
\end{equation}
A single-qubit quantum gate is a unitary operator \(U \in \mathbb{C}^{2 \times 2}\) satisfying \(U^\dagger U = U U^\dagger = I_2\), where \(I_2\) is the \(2 \times 2\) identity matrix. The Pauli operators are three fundamental single-qubit unitaries:
\begin{equation}
    \sigma_X = \begin{pmatrix} 0 & 1 \\ 1 & 0 \end{pmatrix}, \quad
    \sigma_Y = \begin{pmatrix} 0 & -i \\ i & 0 \end{pmatrix}, \quad
    \sigma_Z = \begin{pmatrix} 1 & 0 \\ 0 & -1 \end{pmatrix},
\end{equation}
also denoted as \(X, Y, Z\). The Pauli-\(X\) gate acts as a quantum NOT gate. Other gates, such as the Hadamard gate, have no classical analog:
\begin{equation}
    H = \frac{1}{\sqrt{2}} \begin{pmatrix} 1 & 1 \\ 1 & -1 \end{pmatrix}.
\end{equation}
Single-qubit rotation gates about the \(x\), \(y\), and \(z\) axes are defined as
\begin{equation}
    e^{-i \theta \sigma_X} = \begin{pmatrix} \cos\theta & -i \sin\theta \\ -i \sin\theta & \cos\theta \end{pmatrix}, \quad
    e^{-i \theta \sigma_Y} = \begin{pmatrix} \cos\theta & -\sin\theta \\ \sin\theta & \cos\theta \end{pmatrix}, \quad
    e^{-i \theta \sigma_Z} = \begin{pmatrix} e^{-i\theta} & 0 \\ 0 & e^{i\theta} \end{pmatrix}.
\end{equation}
A sequence of quantum gates forms a single-qubit quantum circuit. A generic single-qubit circuit acting on \(\ket{\psi}\) is shown in \cref{generic-single}.
\begin{figure}[h] \centering \begin{quantikz} \gategroup[wires=1,steps=5,style={rounded corners,draw=none,fill=blue!20}, background]{} \lstick{$\ket{\psi} \; \; $} & \gate{U_1} & \ \cdots\ & \gate{U_m} & \end{quantikz} \caption{A generic single qubit quantum circuit diagram.} \label{generic-single} \end{figure}
Measurement in a quantum circuit produces probabilistic outcomes, with probabilities determined by the state prior to measurement. For a Hermitian operator \(O\) and quantum state \(\ket{\psi}\), the probability of a specific outcome and the expected value are
\begin{align}
    \mathbb{P}(\text{Outcome}) &= |\langle \psi | O | \psi \rangle|^2, \\
    \mathbb{E}(\text{Outcome}) &= \langle \psi | O | \psi \rangle.
\end{align}

	\subsubsection{Multiple Qubit Systems}
A $k$-qubit quantum state is an element of
\begin{equation}
    \mathbb{C}^{2^k} := \underbrace{\mathbb{C}^2 \otimes \cdots \otimes \mathbb{C}^2}_{k \text{ times }},
\end{equation}
for some $k \geq 1$, where $\otimes$ denotes the tensor product. If $\ket{\psi}$ is a single-qubit state, then
\(\ket{\psi}^{\otimes k} \in \mathbb{C}^{2^k}\) represents the $k$-qubit state
$
\ket{\psi}^{\otimes k} = \ket{\psi} \otimes \cdots \otimes \ket{\psi}.
$
The standard computational basis for \(\mathbb{C}^{2^k}\) is
\begin{equation}
    \mathcal{B} = \{ \ket{j_k \cdots j_1} \mid j_i = 0,1 \text{ for } i = 1,\ldots,k \}.
\end{equation}
A $k$-qubit quantum gate is a unitary operator \(U \in \mathbb{C}^{2^k \times 2^k}\) such that \(U^\dagger U = U U^\dagger = I_{2^k}\). An important example of a two-qubit gate is the CNOT gate:
\begin{equation}
    \operatorname{CNOT} = 
    \begin{pmatrix}
        1 & 0 & 0 & 0 \\
        0 & 1 & 0 & 0 \\
        0 & 0 & 0 & 1 \\
        0 & 0 & 1 & 0
    \end{pmatrix}.
\end{equation}
A sequence of $k$-qubit gates forms a $k$-qubit quantum circuit. For such a circuit, the \emph{width} is the number of qubits, and the \emph{depth} is the maximum number of times a qubit is acted upon by a multi-qubit gate. For instance, the three-qubit circuit shown in \cref{3-qubit-circuit} has width 3 and depth 4.

	
	\begin{figure}[h]
		\centering
		\begin{quantikz}
			\gategroup[wires=3,steps=7,style={rounded corners,draw=none,fill=blue!20}, background]{}
			\lstick{$\ket{\psi} \; \; $} & \gate{U_1} & \gate[2]{U_2} &  & \gate{U_5}  & \gate[3]{U_6} &\\
			\lstick{$\ket{\psi} \; \; $} & & & \gate[2]{U_4}  && &  \\
			\lstick{$\ket{\psi} \; \; $} && \gate{U_3} && & & 
		\end{quantikz}
		\caption{A generic $3$-qubit quantum circuit diagram.}
		\label{3-qubit-circuit}
	\end{figure}

	\subsection{Quantum Algorithms}\label{quant-algs}
	We discuss the quantum algorithms that form the basis of this work. \cref{subsection:qsp} covers the quantum signal processing algorithm, and \cref{subsection:lcu} covers the linear combination of unitaries algorithm.
	
	\subsubsection{Quantum Signal Processing}\label{subsection:qsp}
	Quantum signal processing (QSP) is a powerful algorithmic paradigm to design quantum algorithms. QSP was first proposed in \cite{low2017qsp} to solve the quantum simulation problem. Since its inception, QSP has been applied to design quantum algorithms for several tasks in scientific computation. The QSP algorithm constructs $2 \times 2$
	unitary matrices whose entries are complex-valued polynomial functions of a real-valued scalar by taking products of various single qubit rotation and single qubit phase gates. More formally, let $x \in [-1, 1]$ be a scalar with a one-qubit encoding:
	\begin{equation}\label{equation:encoding}
		W(x) := 
		\begin{pmatrix}
			x & i\sqrt{1 - x^2} \\
			i\sqrt{1 - x^2} & x
		\end{pmatrix}
		:= e^{i \arccos(x) \sigma_x},
		\quad \quad
		\theta \in [0,\pi].
	\end{equation}
	For $\boldsymbol \varphi = (\varphi_0, \varphi_1, \ldots, \varphi_\ell) \in \mathbb R^{\ell+1}$, consider
	\begin{equation}\label{equation:qsp-ansatz}
		V_{\boldsymbol \varphi}(x) = e^{i\varphi_0 \sigma_z} W(x) e^{i\varphi_1 \sigma_z} W(x) e^{i\varphi_2 \sigma_z} \ldots W(x) e^{i\varphi_\ell \sigma_z}.
	\end{equation}
	The main result of the QSP algorithm is a complete characterization of which class of polynomial functions can be encoded using the sequence of gates in \eqref{equation:qsp-ansatz}. 
	
	\begin{prop}
		\label{qsp-complex-poly}
		(\cite[Theorem 3]{gilyen2019qsvt})
		Let $\ell \in \mathbb{N}$ and $x \in [-1, 1]$.
		There exists $\boldsymbol \varphi = (\varphi_0, \varphi_1, \ldots, \varphi_\ell) \in \mathbb{R}^{\ell+1}$ such that
		\begin{equation}\label{equation:qsp-result}
			V_{\boldsymbol{\varphi}}(x) = e^{i\varphi_0\sigma_z} \prod_{j=1}^{\ell} \left( W(x) e^{i\varphi_j\sigma_z} \right)
			=
			\begin{pmatrix}
				p(x) & iq(x)\sqrt{1-x^2} \\
				iq^*(x)\sqrt{1-x^2} & p^*(x)
			\end{pmatrix}.
		\end{equation}
		if and only if $p, q \in \mathbb{C}[x]$ such that:
		\begin{enumerate}
			\item[(i)] $\text{deg}(p(x)) \leq \ell$ and $\text{deg}(q(x)) \leq \ell - 1,$
			\item[(ii)] $p(x)$ has parity $\ell \mod 2$ and $q(x)$ has parity $\ell - 1 \mod 2$,
			\item[(iii)] For all  $x \in [-1, 1]$, we have $|p(x)|^2 + (1 - x^2) |q(x)|^2 = 1$.
		\end{enumerate}
	\end{prop}

	\begin{remark}
		A polynomial has parity $0$ if all coefficients corresponding to odd powers of $x$ are $0$. Similarly, a polynomial has parity $1$ if all coefficients corresponding to even powers of $x$ are $0$. 
	\end{remark}

	Note that \cref{qsp-complex-poly} implies that $\bra{0} V_{\boldsymbol \varphi}(x) \ket{0}
	=
	p(x)$. This implies that the expected value of the output of the quantum circuit in \eqref{equation:qsp-ansatz}, as measured in the $\sigma_z$ basis, is the value of the polynomial $p(x)$.

	\begin{remark}
		The limitation of \cref{qsp-complex-poly} is that the achievable polynomial $p(x)$ must be accompanied by another polynomial $q(x)$ satisfying the conditions of \cref{qsp-complex-poly}.
		Fortunately, there exist  variants of the algorithm that can implement real-valued polynomials without being constrained to find a corresponding $q(x)$. However, we will not have to use these variants in this work. 
	\end{remark}

	\subsubsection{Linear Combination of Unitaries}\label{subsection:lcu} 
	The linear combination of unitaries (LCU) technique, introduced in \cite{Childs2012lcu}, provides a framework for realizing a linear combination of unitary operators on a quantum computer. The LCU technique has subsequently developed into a foundational primitive extensively utilized across a wide range of quantum algorithms.  Let \( k, T \in \mathbb{N} \), and let \( a_1, \ldots, a_T \in \R \) and \( U_1, \ldots, U_T \) be unitary operators. The linear combination of unitaries (LCU) algorithm implements the operator given by
	\begin{equation}
		U = \sum_{j=1}^{T} a_j U_j.
	\end{equation}
	The LCU algorithm operates under the assumption that two unitary operators can be implemented. First, it assumes that the state preparation oracle $F$ can be implemented:
	\begin{align}\label{F-LCU}
		F\ket{0} = \frac{1}{\sqrt{\| \boldsymbol{a} \|_1 }} \sum_{j=1}^{T} \sqrt{a_{j}} \ket{j}.
	\end{align}
	Second, it assumes that the following two-qubit controlled gate can be implemented:
	\begin{equation}\label{controlled-LCU}
		U_c = \sum_{j=1}^{T} U_j \otimes \ket{j}\bra{j}.
	\end{equation}
Assuming access to the required unitary operators, the LCU algorithm facilitates the implementation of the desired unitary operator $U_{\text{LCU}} = (I \otimes F^{\dagger} ) U_c (I \otimes F)$. Let \(\ket{\psi}\) denote the input quantum state, and let \(\ket{0}\) represent an ancilla qubit. We have:
\begin{equation}
	U_{\text{LCU}} \ket{\psi} \ket{0} = \frac{1}{\|\boldsymbol{a}\|_1} \left(\sum_{j=1}^{T} a_j U_j\right) \ket{\psi} \ket{0} + \ket{\perp}.
\end{equation}
Here \(\ket{\perp}\) denotes a potentially non-normalized state satisfying \((I \otimes \ket{0}\bra{0}) \ket{\perp} = 0\). We are mainly interested in computing the expected value \(\bra{\psi} U_{\text{LCU}} \ket{\psi}\) for a given quantum state \(\ket{\psi}\). This computation can be performed efficiently using the Hadamard test \cite{1998clevequantumalgorithmsrevisited}. It takes as input the state \(\ket{\psi}\) and the operator \(U_{\text{LCU}}\), producing as output a random variable derived from the measurement of an ancilla qubit, which provides the desired expected value. Specifically, \(\text{Re} \bra{\psi} U_{\text{LCU}} \ket{\psi} \in \mathbb{R}\) can be constructed by executing the quantum circuit illustrated in \cref{figure:lcu-plus-hadamard} (right).
Note that we have
\begin{equation}\label{lcu-output}
	\bra{0} \bra{\psi} U_{\text{LCU}} \ket{\psi} \ket{0} = \frac{1}{\|\boldsymbol{a}\|_1} \bra{\psi} \left(\sum_{j=1}^{T} a_j U_j\right) \ket{\psi}.
\end{equation}

\begin{remark}
	A minor modification of the Hadarmard test can be implemented to output  \( \operatorname{Im} \bra{\psi} \mathrm{U}_{\operatorname{LCU}} \ket{\psi} \in i \R \). 
\end{remark}

\begin{figure}[t]
	\centering
	\sbox0{
		\begin{quantikz}[wire types={b,b},classical gap=0.07cm]
			\gategroup[wires=2,steps=6,style={rounded corners,draw=none,fill=blue!20}, background]{$\text{U}_{\text{LCU}} $}
			\lstick{$\ket{0} \; \; $} &  &  & \gate[2]{U_c} &  & \\
			\lstick{$\ket{0} \; \; $} &   & \gate{F} &  & \gate{F^\dagger} &  
		\end{quantikz}
	}%
	\sbox1{
		\begin{quantikz}[wire types={q,b,b},classical gap=0.07cm]
			\gategroup[wires=3,steps=6,style={rounded corners,draw=none,fill=blue!20}, background]{$\text{Hadamard Test}$}
			\lstick{$\ket{0} \; \; $} & \gate{H}   &  \ctrl{1}  &  \gate{H}  & \meter{}   & \\
			\lstick{$\ket{0} \; \; $} &  & \gate[2]{U_c}  & &  & \\
			\lstick{$\ket{0} \; \; $} & \gate{F}  &  & \gate{F^\dagger} & &  
		\end{quantikz}
	}%
	\begin{tabular}{cc}
		\usebox0 & \usebox1
	\end{tabular}
	\caption{(Left) The quantum circuit implementing the linear combination of unitaries technique.  (Right) 
		The quantum circuit implementing the Hadamard test.
	}
	\label{figure:lcu-plus-hadamard}
\end{figure}

\begin{remark}
	We note that \eqref{lcu-output} implies that the success probability decays as \( \mathcal{O}( 1 / \| \boldsymbol{a} \|_1^2 ) \). In particular, if \( \| \boldsymbol{a} \|_1 = \Omega(c^n) \) for some $c > 1$, then the success probability decays exponentially. In the worst case, to prevent this exponential decay, it is necessary to employ robust oblivious amplitude amplification \cite{childs2015lcutaylor} to boost the success probability to \( 1 - \mathcal{O}(\delta) \) for any \( \delta > 0 \). To achieve this improvement, we must run \( \mathcal{O}( \|\boldsymbol{a}\|_1/\delta) \) rounds of robust oblivious amplitude amplification. However, we will not further address this detailed aspect, as the focus of this work is on the construction of the quantum circuit itself, rather than other concerns such as quantum state readout.
\end{remark}

\section{Main Results}\label{main-results}
We discuss the main results in this section. \cref{cheb-qc} examines the worst-case complexity of quantum circuits required to implement linear combinations of Chebyshev polynomials. \cref{lower-order} builds on this result and provides bounds on the worst-case complexity of quantum circuits required to approximate Korobov functions.

\subsection{Implementing Linear Combination of Chebyshev Polynomials}\label{cheb-qc}
We first discuss how to implement linear combinations of Chebyshev polynomials on a quantum computer. The QSP algorithm is fundamental to constructing quantum circuits for this task since QSP implements univariate Chebyshev polynomials. This is summarized in the lemma below:

\begin{lem}\label{lem:qsp-1d-chebyshev} 
	Let $T_r \in \mathbb{R}[x]$ be the degree-$r$ Chebyshev polynomial of the first kind. Let $\boldsymbol \varphi \in \mathbb{R}^{r+1}$ such that $\varphi_i = 0$
	for all $i = 0, \cdots, r$.  For this specific choice $\boldsymbol \varphi$, there exists an unparameterized quantum circuit, $U_r(x)$, such that 
	\begin{equation}
		\bra{0} U_r(x) \ket{0}  = T_r(x)
	\end{equation}
	for each $r \in \mathbb N \cup \{0\}$. 
	The quantum circuit has width $1$, depth $2r+1$, and $r+1$ (predetermined) parameters.
\end{lem}

\begin{remark}
	Reference \cite{gilyen2019qsvt} proves a slightly modified version of \cref{lem:qsp-1d-chebyshev}. This discrepancy arises because \cite{gilyen2019qsvt} uses a slightly different quantum circuit for quantum signal processing than the one we use. For completeness, we prove the result below.
\end{remark}

\begin{proof}
	(\cref{lem:qsp-1d-chebyshev})
	The claim is true for $r=0$ since in this case $U_r(x) = I$ and the $(1,1)$ entry of $I$ is $T_0(x)=1$. We prove by induction that
	\begin{equation}
		W^r(x) = \begin{pmatrix}
			T_r(x) & i\sqrt{1-x^2}S_{r-1}(x) \\
			i\sqrt{1-x^2}S_{r-1}(x) & T_r(x)\end{pmatrix},\end{equation}
	for each $r \geq 1$. For $r=1$, we have
	\begin{equation}
		W(x) 
		=
		\begin{pmatrix}
			x & i\sqrt{1-x^2} \\
			i\sqrt{1-x^2} & x
		\end{pmatrix}.
	\end{equation}
	Since $T_1(x) = x$ and $S_0(x)=1$, the claim is true. Now assume the claim is true for some $r > 1$. We have,
	\begin{align}
		W^{r+1}(x) &= W^r(x) W(x)  \\
		&= \begin{pmatrix}
			T_r(x) & i\sqrt{1-x^2}S_{r-1}(x) \\
			i\sqrt{1-x^2}S_{r-1}(x) & T_r(x)
		\end{pmatrix} 
		\begin{pmatrix}
			x & i\sqrt{1-x^2} \\
			i\sqrt{1-x^2} & x
		\end{pmatrix}  \\
		&= \begin{pmatrix}
			xT_r(x) - (1-x^2)S_{r-1}(x) & i\sqrt{1-x^2}(T_r(x) +xS_{r-1}(x)) \\
			i\sqrt{1-x^2}(T_r(x) +xS_{r-1}(x)) & xT_r(x) - (1-x^2)S_{r-1}(x)
		\end{pmatrix}  \\
		& = 
		\begin{pmatrix}
			T_{r+1}(x) & i\sqrt{1-x^2}S_{r}(x) \\
			i\sqrt{1-x^2}S_{r}(x) & T_{r+1}(x)\end{pmatrix}. 
	\end{align}
	The last equality follows since the Chebyshev polynomials follows that they also satisfy a pair of mutual recurrence equations:
	\begin{align}
		T_{r+1}(x) &=xT_r(x) - (1-x^2)S_{r-1}(x), \\
		S_{r}(x) &= (x)T_r(x) + xS_{r-1}(x).
	\end{align}
	Since we have
	\begin{equation}
		U_r(x) = e^{i0\sigma_z} \prod_{j=1}^{r} \left( W(x) e^{i 0 \sigma_z} \right) = W^r(x) = \begin{pmatrix}
			T_r(x) & i\sqrt{1-x^2}S_{r-1}(x) \\
			i\sqrt{1-x^2}S_{r-1}(x) & T_r(x)
		\end{pmatrix},
	\end{equation}
	the claim follows. Clearly, $U_r(x)$ has  width $1$, depth $2r+1$, and $r+1$ (predetermined) parameters. See \cref{figure:Chebyshev-qsp} for the corresponding quantum circuit.
\end{proof}

\begin{figure}[h]
	\centering
	\begin{quantikz}
		\gategroup[wires=1,steps=8,style={rounded corners,draw=none,fill=blue!20}, background]{}
		\lstick{$\ket{0} \; \; $} & \gate{I} & \gate{W(x)}  & \gate{I}  &  \ \cdots\ &  \gate{W(x)} & \gate{I} &
	\end{quantikz}
	\caption{The circuit diagram for the quantum circuit $U_r(x)$ in \cref{lem:qsp-1d-chebyshev}.}
	\label{figure:Chebyshev-qsp}
\end{figure}

We aim to implement linear combinations of products of Chebyshev polynomials on a quantum computer. To achieve this, we leverage the Linear Combination of Unitaries (LCU) technique. Since Chebyshev polynomials can be implemented as unitary operators via quantum signal processing (QSP), the LCU framework provides a natural and efficient route to realize linear combinations of them on a quantum computer.

\begin{remark}\label{measurement-not}
	In what follows, for ease of notation, we will write a multi-qubit system initialized in the state \( \ket{0}^{\otimes M} \) for some \( M \geq 1 \) simply as \( \ket{0} \). Moreover, we will denote the multi-qubit unitary operator $\sigma_z \otimes I \otimes \cdots \otimes I$ as simply \( Z^{(1)} \). Here, \( Z^{(1)} \) represents the Pauli \( \sigma_z \) observable which is measured only on the first qubit.
\end{remark}

\begin{prop}\label{pqc-lcu-cheb}
	Fix $d,M \in \mathbb N$.
	Let $x \in [-1,1]^d$ and consider a linear combination of products of Chebyshev polynomials:
	\begin{equation}
		f(\boldsymbol x) = \sum_{i=1}^M a_i \prod_{j=1}^d T_{n_{ij}}(x_j).
	\end{equation}
	Here $a_1, \cdots, a_M \in \mathbb R$ and $n_{i,j} \in \mathbb N $ for $i=1, \cdots, N$ and $j=1,\cdots, d$.
	There exists a quantum circuit \( U_{\operatorname{Cheb}}(\boldsymbol{x}) \) such that if a measurement corresponding to the Pauli \( \sigma_z \) observable is made only on the first qubit, then
	\begin{equation}
		\bra{0} U_{\operatorname{Cheb}}^\dagger(\boldsymbol x) Z^{(1)} U_{\operatorname{Cheb}}(\boldsymbol x)  \ket{0} = f(\boldsymbol{x}).
	\end{equation}
	The quantum circuit $U_{\operatorname{Cheb}}(\boldsymbol x)$ has depth \(\mathcal{O}( (2\|\boldsymbol{n}\|_1  + dM) \log_2 M)\) and width \(\mathcal{O}(d + \log_2 M)\).
\end{prop}

\begin{proof}
	For each $n_{i,j} \in \mathbb N$,
	\cref{lem:qsp-1d-chebyshev} implies that there exists a quantum circuit $U_{n_{i,j}}(x_j)$ such that
	$\bra{0} U_{n_{i,j}}(x_j) \ket{0} = T_{n_{ij}}(x_j)$. If we consider the unitary operator
	\begin{equation}
		U_{i}(\boldsymbol x) = \bigotimes_{j=1}^d U_{n_{i,j}} (x_j),
	\end{equation} 
	then $U_{i}(\boldsymbol x)$ can be implemented with a width $d$ quantum circuit comprising of $2\|\boldsymbol{n}_i\|_1+d$ quantum gates. Here we have defined $\boldsymbol{n}_i=(n_{i1},\cdots,n_{id})$. 
	We now aim to use the LCU technique to implement the following unitary operator:
	\begin{equation}\label{equation:lcu-operator}
		U_{\operatorname{Cheb}}(\boldsymbol x) := U_{\operatorname{LCU}}(\boldsymbol x) = \sum_{i=1}^N a_i U_{i}(\boldsymbol x). 
	\end{equation}

	\begin{enumerate}
		\item[(i)] First consider the unitary operator
		\begin{equation}\label{equation:state-preperation}
			F\ket{0} = \frac{1}{\sqrt{\| \boldsymbol a \|_1}}
			\sum_{j=1}^{N} \sqrt{a_{\boldsymbol{i}}}\ket{j}.
		\end{equation}
		Without loss of generality, assume that each coefficient \(a_{\boldsymbol{i}}\) is strictly positive. If a coefficient is negative, its sign can be absorbed into the corresponding unitary operator \(U_{i}(\boldsymbol{x})\) without affecting the construction. Since there are $M$ terms in the sum in \cref{equation:lcu-operator}, the function \(F\) can be implemented using a quantum circuit with \(\mathcal{O}(M)\) gates acting on \(\mathcal{O}(\log_2 M)\) ancilla qubits.

		\item[(ii)] Next consider the controlled unitary operator 
		\begin{equation}
			U_c(\boldsymbol x) = \sum_{i=1}^{N} U_{i}(\boldsymbol x) \otimes \ket{i}\bra{i}.
		\end{equation}
		This operator \(U_c(\boldsymbol x)\) acts on a composite quantum system that includes \(\mathcal{O}(d)\) computational qubits and \(\mathcal{O}(\log_2 M)\) ancilla qubits. 
		Since each \(U_{i}(\boldsymbol x)\) can be implemented using a \(2\|\boldsymbol{n}_i\|_1+d\) single-qubit gates, \(U_c(\boldsymbol x)\) can be implemented via
		\begin{equation}
			\sum_{i=1}^M (2\|\boldsymbol{n}_i\|_1+d)
			= \sum_{i=1}^M \sum_{j=1}^d 2 n_{ij} +d M
			:= 2\|\boldsymbol{n}\|_1  + dM
		\end{equation}
		single qubit gates controlled on \(\mathcal{O}( \log_2 M )\) ancilla qubits. A $\mathcal{O}(\log_2 M)$-qubit controlled gate acting can be implemented via a quantum circuit using only CNOT gates and single-qubit gates such that the circuit has depth \(\mathcal{O}(\log_2 M)\) \cite{daSilva2022LinearDepthMultiQubit, yu2024provable}. Consequently, the overall implementation of \(U_c(\boldsymbol x)\) involves a quantum circuit with depth \(\mathcal{O}( (2\|\boldsymbol{n}\|_1  + dM) \log_2 M)\) and width \(\mathcal{O}(d + \log_2 M)\). 
		
		\item[(iii)] Next, we leverage the LCU algorithm to implement the operator \(U_{\text{LCU}}\) defined as 
        \begin{equation}
		U_{\text{LCU}}(\boldsymbol{x}) = (I \otimes F^{\dagger}) U_c(\boldsymbol{x}) (I \otimes F).
        \end{equation}
        The asymptotic scaling of both the depth and the width for implementing \(U_{\text{LCU}}\) is dominated by the complexity of \(U_c(\boldsymbol x)\).
	\end{enumerate}
	The desired quantum circuit, \( U_{\operatorname{Cheb}}(\boldsymbol{x}) \), is constructed by applying the Hadamard test to a single ancilla qubit in order to estimate the expected value of \( U_{\operatorname{LCU}} \). Note that we have
	\begin{equation}
		\bra{0}  U_{\operatorname{LCU}}(\boldsymbol x) \ket{0} 
		= \sum_{i=1}^M a_i \prod_{j=1}^d \bra{0} U_{n_{ij}} \left( x_j \right) \ket{0} 
		= \sum_{i=1}^M a_i \prod_{j=1}^d T_{n_{ij}}(x_j).
	\end{equation}
	The additional operations involved in the Hadamard test have \( \mathcal{O}(1) \) complexity. Thus, the overall complexity of the algorithm is determined by the implementation of the LCU algorithm. This completes the proof.
\end{proof}

The significance of \cref{pqc-lcu-cheb} is that we have an explicit quantum circuit along with appropriate complexity estimates to implement an arbitrary 
linear combination of Chebyshev polynomials on a quantum computer. We will leverage this result in our analysis below.

\subsection{Approximating Order Korobov Functions}\label{lower-order}
We begin by describing the hierarchical basis construction that underlies the sparse grid decomposition. We then introduce the Korobov function space and provide an estimate of the quantum circuit complexity required to approximate low-order Korobov functions.

\subsubsection{Hierarchical Basis}
We first consider the one-dimensional case (\(d = 1\)). Fix \(h_n = 2^{-n} \in \mathbb{R}\) for some \(n \in \mathbb{N}\). Many numerical methods rely on uniform discretizations of the domain. Let \(\mathcal{P}_n\) denote a uniform discretization of \([0,1]\) consisting of the points \(\{ i h_n \}_{i=1}^{2^n - 1}\). We now define a piecewise linear function, referred to as the \emph{hat function}:
\begin{equation}\label{hat-func}
	\phi(x) = 
	\begin{cases} 
		1 - |x|, & \text{if } x \in [-1, 1], \\
		0, & \text{otherwise.}
	\end{cases}
\end{equation}
The approach is to consider a family of finite-dimensional vector spaces of one-dimensional hat functions, which can be used to approximate more complex functions on \([0,1]\). Specifically, for each \(x_i \in \mathcal{P}_n\), we define \(\phi_{n,i}(x)\) by
\begin{equation}
	\phi_{n,i}(x) = \phi\left( \frac{x - x_i}{h_n} \right).
\end{equation}
This function \( \phi_{n,i}(x) \) is a scaled and translated version of the hat function \( \phi(x) \), centered at \( x_i \) with width \( h_n \). We now consider the finite-dimensional vector, \( V_n \), defined as follows:
\begin{equation}
	V_n = \text{span}\{\phi_{n,i} : 1 \leq i \leq 2^n - 1\}.
\end{equation}
This space \( V_n \) forms our approximating space. The basis functions for $V_n $ are called the 1-dimensional nodal basis (\cref{figure:nodal-basis}).
\begin{figure}[t]
	\centering
	\[
	\begin{tikzpicture}[scale=0.8]
		\begin{axis}[
			axis lines = center,
			xmin=-0.2, xmax=1.2,
			ymin=-0.2, ymax=1.2,
			xtick={0},
			ytick=\empty, 
			xlabel style={below right},
			ylabel style={above left},
			minor tick num=1,
			width=10cm,
			height=6cm,
			axis line style={thick, black}, 
			]
			
			\node[circle, fill=black, inner sep=1.5pt] at (axis cs: 0/8, 0) {};
			\node[circle, fill=black, inner sep=1.5pt] at (axis cs: 1/8, 0) {};
			\node[circle, fill=black, inner sep=1.5pt] at (axis cs: 2/8, 0) {};
			\node[circle, fill=black, inner sep=1.5pt] at (axis cs: 3/8, 0) {};
			\node[circle, fill=black, inner sep=1.5pt] at (axis cs: 4/8, 0) {};
			\node[circle, fill=black, inner sep=1.5pt] at (axis cs: 5/8, 0) {};
			\node[circle, fill=black, inner sep=1.5pt] at (axis cs: 6/8, 0) {};
			\node[circle, fill=black, inner sep=1.5pt] at (axis cs: 7/8, 0) {};
			\node[circle, fill=black, inner sep=1.5pt] at (axis cs: 8/8, 0) {};
			
			\addplot[domain=0/8:2/8, samples=100, blue, thick] {1 - abs(8*(x - 1/8))} node[pos=0.5, above, black] {\(\phi_{3,1}\)};
			\addplot[domain=1/8:3/8, samples=100, blue, thick] {1 - abs(8*(x - 2/8))} node[pos=0.5, above, black] {\(\phi_{3,2}\)};
			\addplot[domain=2/8:4/8, samples=100, blue, thick] {1 - abs(8*(x - 3/8))} node[pos=0.5, above, black] {\(\phi_{3,3}\)};
			\addplot[domain=3/8:5/8, samples=100, blue, thick] {1 - abs(8*(x - 4/8))} node[pos=0.5, above, black] {\(\phi_{3,4}\)};
			\addplot[domain=4/8:6/8, samples=100, blue, thick] {1 - abs(8*(x - 5/8))} node[pos=0.5, above, black] {\(\phi_{3,5}\)};
			\addplot[domain=5/8:7/8, samples=100, blue, thick] {1 - abs(8*(x - 6/8))} node[pos=0.5, above, black] {\(\phi_{3,6}\)};
			\addplot[domain=6/8:8/8, samples=100, blue, thick] {1 - abs(8*(x - 7/8))} node[pos=0.5, above, black] {\(\phi_{3,7}\)};
			
		\end{axis}
	\end{tikzpicture}
	\]
	\caption{The $1$-dimensional nodal basis for $V_3$ consists of functions $\phi_{3,i}$, $1 \leq i \leq 7$.
	}
	\label{figure:nodal-basis}
\end{figure}
The sparse grid construction, discussed in the next subsection, exploits the fact that \(V_n\) can be defined as a sequence of nested approximation spaces, each corresponding to progressively refined grids. The key idea is to select grid points hierarchically. Consider a sequence of nested grids \(\{G_\ell\}_{\ell=1}^n\), where the level-\(k\) grid is defined by
\begin{equation}
	G_\ell := \left\{ x_{i}^{(\ell)} = \frac{i}{2^\ell} : i = 1, 3, 5, \dots, 2^\ell - 1 \right\}.
\end{equation}
Note that only the odd-indexed points are taken. This ensures that \( G_{\ell} \cap G_{\ell'} = \emptyset \) for \( \ell \neq \ell' \), enabling a hierarchical decomposition. Note that we have
\begin{equation}
	\mathcal{P}_n = \coprod_{\ell=1}^n G_\ell,
\end{equation}
which avoids redundancy and supports a multi-resolution representation. Note that if we consider the finite-dimensional vector space
\begin{equation}
	W_\ell = \text{span}\{\phi_{\ell,i} : i=1,3,5\cdots, 2^\ell - 1 \},
\end{equation}
it follows that the space \( V_n \) is the direct sum \begin{equation}
	V_n = \bigoplus_{1 \leq \ell \leq n} W_\ell.
\end{equation}
The basis functions of \( W_\ell \) are referred to as the 1-dimensional hierarchical basis functions (\cref{figure:hierarchical-basis}). The 1-dimensional hierarchical basis functions capture the variations in resolution between successive levels.

\begin{figure}[h]
	\centering
	\begin{subfigure}[b]{0.32\linewidth} 
		\centering
		\begin{tikzpicture}[scale=0.45]
			\begin{axis}[
				axis lines = center,
				xmin=-0.1, xmax=1.1,
				ymin=-0.1, ymax=1.2,
				xtick={0},
				ytick=\empty,
				xlabel style={below right},
				ylabel style={above left},
				minor tick num=1,
				width=10cm,
				height=6cm,
				axis line style={thick, black},
				]
				\node[circle, fill=black, inner sep=1.5pt] at (axis cs: 0, 0) {};
				\node[circle, fill=black, inner sep=1.5pt] at (axis cs: 1/2, 0) {};
				\node[circle, fill=black, inner sep=1.5pt] at (axis cs: 1, 0) {};
				
				\addplot[domain=0:1, samples=100, blue, ultra thick] {1 - abs(2*(x - 1/2))} node[pos=0.5, above, black] {\(\phi_{1,1}\)};
			\end{axis}
		\end{tikzpicture}
		\caption{$W_1$}
		\label{hierarchical-basis:(a)}
	\end{subfigure}
	\quad
	\begin{subfigure}[b]{0.3\linewidth} 
		\centering
		\begin{tikzpicture}[scale=0.45]
			\begin{axis}[
				axis lines = center,
				xmin=-0.1, xmax=1.1,
				ymin=-0.1, ymax=1.2,
				xtick={0},
				ytick=\empty,
				xlabel style={below right},
				ylabel style={above left},
				minor tick num=1,
				width=10cm,
				height=6cm,
				axis line style={thick, black},
				]
				\node[circle, fill=black, inner sep=1.5pt] at (axis cs: 0/4, 0) {};
				\node[circle, fill=black, inner sep=1.5pt] at (axis cs: 1/4, 0) {};
				\node[circle, fill=black, inner sep=1.5pt] at (axis cs: 2/4, 0) {};
				\node[circle, fill=black, inner sep=1.5pt] at (axis cs: 3/4, 0) {};
				\node[circle, fill=black, inner sep=1.5pt] at (axis cs: 4/4, 0) {};
				
				\addplot[domain=0/4:2/4, samples=100, blue, ultra thick] {1 - abs(4*(x - 1/4))} node[pos=0.5, above, black] {\(\phi_{2,1}\)};
				\addplot[domain=2/4:4/4, samples=100, blue, ultra thick] {1 - abs(4*(x - 3/4))} node[pos=0.5, above, black] {\(\phi_{2,3}\)};
			\end{axis}
		\end{tikzpicture}
		\caption{$W_2$}
		\label{hierarchical-basis:(b)}
	\end{subfigure}
	\quad
	\begin{subfigure}[b]{0.3\linewidth} 
		\centering
		\begin{tikzpicture}[scale=0.45]
			\begin{axis}[
				axis lines = center,
				xmin=-0.1, xmax=1.1,
				ymin=-0.1, ymax=1.2,
				xtick={0},
				ytick=\empty,
				xlabel style={below right},
				ylabel style={above left},
				minor tick num=1,
				width=10cm,
				height=6cm,
				axis line style={thick, black},
				]
				\node[circle, fill=black, inner sep=1.5pt] at (axis cs: 0/8, 0) {};
				\node[circle, fill=black, inner sep=1.5pt] at (axis cs: 1/8, 0) {};
				\node[circle, fill=black, inner sep=1.5pt] at (axis cs: 2/8, 0) {};
				\node[circle, fill=black, inner sep=1.5pt] at (axis cs: 3/8, 0) {};
				\node[circle, fill=black, inner sep=1.5pt] at (axis cs: 4/8, 0) {};
				\node[circle, fill=black, inner sep=1.5pt] at (axis cs: 5/8, 0) {};
				\node[circle, fill=black, inner sep=1.5pt] at (axis cs: 6/8, 0) {};
				\node[circle, fill=black, inner sep=1.5pt] at (axis cs: 7/8, 0) {};
				\node[circle, fill=black, inner sep=1.5pt] at (axis cs: 8/8, 0) {};
				
				\addplot[domain=0/8:2/8, samples=100, blue, ultra thick] {1 - abs(8*(x - 1/8))} node[pos=0.5, above, black] {\(\phi_{3,1}\)};
				\addplot[domain=2/8:4/8, samples=100, blue, ultra thick] {1 - abs(8*(x - 3/8))} node[pos=0.5, above, black] {\(\phi_{3,3}\)};
				\addplot[domain=4/8:6/8, samples=100, blue, ultra thick] {1 - abs(8*(x - 5/8))} node[pos=0.5, above, black] {\(\phi_{3,5}\)};
				\addplot[domain=6/8:8/8, samples=100, blue, ultra thick] {1 - abs(8*(x - 7/8))} node[pos=0.5, above, black] {\(\phi_{3,7}\)};
			\end{axis}
		\end{tikzpicture}
		\caption{$W_3$}
		\label{hierarchical-basis:(c)}
	\end{subfigure}
	\caption{The $1$-dimensional hierarchical basis vectors for $W_1, W_2$ and $W_3$ are plotted.}
	\label{figure:hierarchical-basis}
\end{figure}

We now extend to higher dimensions (\(d>1\)) via a tensor product construction. Let \(\boldsymbol{h}_{\boldsymbol{n}} = 2^{-\boldsymbol{n}}\) for \(\boldsymbol{n} \in \mathbb{N}^d\) and consider \(\boldsymbol{n}=(n,\ldots,n)\) for \(n \in \mathbb{N}\). A uniform discretization of \([0,1]^d\) is \(\mathcal{P}_{\boldsymbol{n}} = \{ \boldsymbol{i} \cdot \boldsymbol{h}_{\boldsymbol{n}} \}_{\boldsymbol{i}=\boldsymbol{1}}^{2^{\boldsymbol{n}}-\boldsymbol{1}}\), and the \(d\)-dimensional \emph{hat functions} are defined as products of the one-dimensional basis functions.
\begin{equation}
	\phi_{\boldsymbol{n}, \boldsymbol{i}}(\boldsymbol{x}) 
	= \prod_{j=1}^{d} \phi_{n_j,i_j}(x_j) 
	= \prod_{j=1}^{d} \phi \left( \frac{x_j - x_{i_j}}{h_{n_j}} \right).
\end{equation}
We define the function spaces \( V_{\boldsymbol{n}} \) and \( W_{\boldsymbol{\ell}} \) as follows:
\begin{align}
	V_{\boldsymbol{n}} &= \operatorname{span}\left\{ \phi_{\boldsymbol{n}, \boldsymbol{i}} : \boldsymbol{1} \leq \boldsymbol{i} \leq 2^{\boldsymbol{n}} - \boldsymbol{1} \right\}, \\
	W_{\boldsymbol{\ell}} &= \operatorname{span}\left\{ \phi_{\boldsymbol{\ell}, \boldsymbol{i}} : \boldsymbol{1} \leq \boldsymbol{i} \leq 2^{\boldsymbol{\ell}} - \boldsymbol{1},\; i_j \text{ odd for all } j \right\}.
\end{align}
Here, \(V_{\boldsymbol{n}}\) denotes the \(d\)-dimensional nodal basis, while \(W_{\boldsymbol{\ell}}\) defines the \(d\)-dimensional hierarchical basis. The hierarchical basis functions in \(W_{\boldsymbol{\ell}}\) capture the incremental detail at resolution level \(\boldsymbol{\ell}\) and are supported on coarser or finer grids depending on \(\boldsymbol{\ell}\). These spaces satisfy the orthogonal decomposition:
\begin{equation}
	V_{\boldsymbol{n}} = \bigoplus_{\| \boldsymbol{\ell} \|_{\infty} \leq \boldsymbol{n}} W_{\boldsymbol{\ell}}.
\end{equation}

\subsubsection{Korobov Function Space}
We now introduce the Korobov function space, which will be shown to admit approximation by \(d\)-dimensional hierarchical basis functions. For \(2 \leq p \leq \infty\), the Korobov function space is defined as follows:
\begin{equation}
	X^{2,p}([0,1]^d) = \left\{ f \in L^p([0,1]^d) : f|_{\partial[0,1]^d} = 0, \, D^{\boldsymbol k} f \in L^p([0,1]^d), \, \|\boldsymbol k\|_{\infty} \leq 2 \right\}.
\end{equation}
Note that \( X^{2,p}([0,1]^d) \) is a subspace of the Sobolev space \( W^{2,p}([0,1]^d) \). For example, in two dimensions $(d=2)$ a function, $f$, belongs to \( X^{2,p}([0,1]^d) \) if and only if
\begin{equation}
	\frac{\partial f}{\partial x_1},\quad \frac{\partial f}{\partial x_2},\quad 
	\frac{\partial^2 f}{\partial x_1^2},\quad \frac{\partial^2 f}{\partial x_2^2},\quad 
	\frac{\partial^2 f}{\partial x_1 \partial x_2},\quad 
	\frac{\partial^3 f}{\partial x_1^2 \partial x_2},\quad 
	\frac{\partial^3 f}{\partial x_1 \partial x_2^2},\quad 
	\frac{\partial^4 f}{\partial x_1^2 \partial x_2^2} \in L^p([0,1]^d).
\end{equation}
Whereas $f$ belongs to \( W^{2,p}([0,1]^d) \) if and only if
\begin{equation}
	\frac{\partial f}{\partial x_1},\quad \frac{\partial f}{\partial x_2},\quad 
	\frac{\partial^2 f}{\partial x_1^2},\quad \frac{\partial^2 f}{\partial x_2^2},\quad 
	\frac{\partial^2 f}{\partial x_1 \partial x_2} \in L^p([0,1]^d).
\end{equation}
The difference arises from the multi-index conditions in the definitions of \( X^{2,p}([0,1]^d) \) and \( W^{2,p}([0,1]^d) \). Specifically, the Korobov space corresponds to the condition \( \|\boldsymbol{k}\|_\infty \leq 2 \), while the Sobolev space corresponds to \( \|\boldsymbol{k}\|_1 \leq 2 \). 

An important property of the Korobov space is that any function \( f \in X^{2,p}([0,1]^d) \) admits an infinite expansion in terms of the hierarchical basis functions. Indeed, it is a standard result that the function space
\begin{equation}
	V = \bigoplus_{\boldsymbol \ell \in \mathbb{N}^d} W_{\boldsymbol \ell}
\end{equation}
is dense in \( H^1_0([0,1]^d) \); that is, \( \overline{V} = H^1_0([0,1]^d) \), where the overline denotes the closure of \( V \) in the \( H^1 \)-norm. The inclusion
\begin{equation}
	X^{2,p}([0,1]^d) \subseteq H^1_0([0,1]^d)
\end{equation}
implies that any function \( f \in X^{2,p}([0,1]^d) \) admits an infinite expansion in terms of the hierarchical basis functions. Specifically,
\begin{equation}\label{equation:infinite-expan}
	f(\boldsymbol{x}) = \sum_{\boldsymbol{\ell}} \sum_{\boldsymbol{i} \in I_{\boldsymbol{\ell}}} v_{\boldsymbol{\ell}, \boldsymbol{i}} \, \phi_{\boldsymbol{\ell}, \boldsymbol{i}}(\boldsymbol{x}),
\end{equation}
where the index set \( I_{\boldsymbol{\ell}} \) is defined as
\begin{equation}
	I_{\boldsymbol{\ell}} = \left\{ \boldsymbol{i} \in \mathbb{N}^d \,\middle|\, \boldsymbol{1} \leq \boldsymbol{i} \leq 2^{\boldsymbol{\ell}} - \boldsymbol{1},\ \text{and } i_j \text{ is odd for each } j \right\}.
\end{equation}
In the remainder of this section, we focus on the case where \( p \in \{2, \infty\} \). This is essentially because the so-called optimal sparse grid decomposition is discussed in \cite{bungartz2004sparse} only for the case \( p = 2, +\infty \).  We have the following result:

\begin{prop}\label{bungartz2004sparse}
	(\cite[Lemma 3.3]{bungartz2004sparse})
	Let $p \in \{2,\infty\}$ and consider $X^{2,p}([0,1]^d)$ endowed with the semi-norm on $X^{2,p}([0,1]^d)$:
	\begin{equation}
		| f |_{\boldsymbol 2,p} := \left \| {\frac{\partial^{2d} f}{\partial x^2_1 \cdots \partial x^2_d}} \right \|_{L^p([0,1]^d)}.
	\end{equation}
	Each $f \in X^{2,p}([0,1]^d)$ can be written as in \eqref{equation:infinite-expan} such that
	\begin{equation}\label{equation:coefficients-integral-formula}
		v_{\boldsymbol \ell,\boldsymbol i} = \int_{[0,1]^d} 
		\prod_{j=1}^{d} \left( -2^{-(\ell_j+1)} \phi_{\ell_j,i_j}(x_j) \right) \frac{\partial^{2d} f}{\partial x_1^2 \cdots \partial x_d^2}(\boldsymbol x) \, d \boldsymbol x
	\end{equation}
	with the following bound on the coefficients:
	\begin{enumerate}
		\item[(i)] If $p = 2$, then
		\begin{equation}\label{equation:coeff-bound-L2}
			|v_{\boldsymbol \ell,\boldsymbol i}| \leq 
			2^{-d} \left ( \frac{2}{3} \right )^{d/2} 2^{- (3/2) \cdot \| \boldsymbol{\ell} \|_1 } | f \mathds{1}_{\operatorname{supp}(\phi_{\boldsymbol{\ell},\boldsymbol{i}})} |_{ \boldsymbol{2}, 2}.
		\end{equation}
		\item[(ii)] If $p = \infty$, then
		\begin{equation}\label{equation:coeff-bound-Linfinity}
			|v_{\boldsymbol \ell,\boldsymbol i}| \leq 2^{-d - 2|\boldsymbol \ell|_1} 
			| f |_{\boldsymbol 2,\infty},
		\end{equation}
		where $\operatorname{supp}(\phi_{\boldsymbol{\ell},\boldsymbol{i}})$ is the support of $\phi_{\boldsymbol{\ell},\boldsymbol{i}}$.
	\end{enumerate}
\end{prop}

Note that the coefficients \(v_{\boldsymbol{\ell}, \boldsymbol{i}}\) exhibit exponential decay in the level index, satisfying \(\|v_{\boldsymbol{\ell}, \boldsymbol{i}}\| = \mathcal{O}(2^{-2|\boldsymbol{\ell}|_1})\). This behavior arises because a function in the Korobov space involves mixed second-order derivatives in all spatial directions. In contrast, the dimension of the subspace \(W_{\boldsymbol{\ell}}\) at level \(\boldsymbol{\ell}\) grows exponentially as \(\mathcal{O}(2^{|\boldsymbol{\ell}|_1})\). An optimization procedure, as described in \cite{bungartz2004sparse}, balances the trade-off between the number of degrees of freedom and the resulting approximation error. This leads to the sparse grid decomposition of level \(n \in \mathbb{N}\):
\begin{equation}\label{equation:sparse-grid}
    V^{s}_{n} = \bigoplus_{\| \boldsymbol{\ell} \|_1 \leq n+d-1} W_{\boldsymbol{\ell}}.
\end{equation}
The sparse grid decomposition has the following properties:

\begin{prop}\label{prop:sparse-grid-korobov}
	Let $d,n \in \mathbb N, \; \epsilon > 0$ and $p \in \{2, \infty\}$. Consider the sparse grid decomposition as in \eqref{equation:sparse-grid}. Then:
	\begin{enumerate} 
		\item[(i)] (\cite[Lemma 3.6]{bungartz2004sparse}) The number of grid points, $N$, is $N = \mathcal{O}(2^n n^{d-1})$.
		\item[(ii)] (\cite[Lemma 3.13]{bungartz2004sparse}) Any $f \in X^{2,p}([0,1]^d)$ can be approximated by
		\begin{equation}\label{sparse-approx-poly}
			f_n^{s} (\boldsymbol{x}) = \sum_{\substack{ \| \boldsymbol{\ell} \|_1 \leq n + d - 1 \\ \boldsymbol{i} \in {I}_{\boldsymbol{\ell}}}}v_{\boldsymbol{\ell},\boldsymbol{i}} \phi_{\boldsymbol \ell,\boldsymbol i}(\boldsymbol x) \; \in \; V^{s}_{n}
		\end{equation}
		such that with respect to the $\| \cdot \|_{L^p([0,1]^d)}$ norm, we have
		\begin{equation}
			\| f - f^{s}_n \|_{L^p([0,1]^d)} = \mathcal{O}(N^{-2} \log^{3(d-1)}_2 N).
		\end{equation}
		For every $\epsilon \in (0,1)$, it suffices to choose 
		$N = \mathcal{O}(\epsilon^{-\frac{1}{2}} \log_2 (1/\epsilon)^{\frac{3}{2}(d-1)})$ in order to ensure that
		\begin{equation}
			\| f - f^{s}_n \|_{L^p([0,1]^d)} \leq \epsilon.
		\end{equation}
	\end{enumerate}
\end{prop}

Traditional numerical methods that discretize each coordinate axis using \( n \) points result in a total of \( \mathcal{O}(n^d) \) grid points. Sparse grid techniques address this challenge by selecting a structured subset of the full tensor-product basis, thereby significantly reducing the number of degrees of freedom with only a modest loss in accuracy, as made precise in \cref{prop:sparse-grid-korobov}.


\subsubsection{Approximation Rates}
We now present our main result, which provides a rigorous complexity analysis of the proposed quantum circuit for approximating functions in the Korobov space \(X^{2,p}([0,1]^d)\). We begin with the cases \(p \in \{2, \infty\}\). Recall that each one-dimensional hierarchical basis function is obtained by scaling and translating the fundamental \emph{hat function} \(\phi(x)\). A key observation, which is both simple and central to our construction, is that the function \(1 - |x|\) can be expressed as a linear combination of Chebyshev polynomials of degrees 0 and 1. Specifically, for \(x \in [-1,1]\), we have
\begin{equation}
	1 \pm x = P_0(x) \pm P_1(x),
\end{equation}
where \( P_0 \) and \( P_1 \) are the Chebyshev polynomials of the first kind of degree \( 0 \) and \( 1 \), respectively. Using this identity, we can express the hat function \( \phi(x) \) as:
\begin{equation}\label{equation:hat-decomposition-chebyshev}
	\phi(x) = [P_0(x) + P_1(x)] \mathds{1}_{[-1,0]}(x)  + [P_0(x) - P_1(x)] \mathds{1}_{[0,1]}(x).
\end{equation}
Based on the discussion in the preceding subsection, any function \( f \in X^{2,p}([0,1]^d) \), for \( p \in \{2, \infty\} \), admits an approximation of the form \( f_n^{s}(\boldsymbol{x}) \) as defined in \eqref{sparse-approx-poly}. For each \( \boldsymbol{x} \in [0,1]^d \), we have
\begin{align}\label{align:function-to-be-approximated}
	f_n^{s}(\boldsymbol x) & =
	\sum_{\substack{\| \boldsymbol{\ell} \|_1 \leq n + d - 1 \\ \boldsymbol{i} \in {I}_{\boldsymbol{\ell}}}} 
	\sum_{k_1, \cdots, k_d \in \{0,1\} }
	(-1)^{\sum_{j=1}^d k_j ( \text{sgn} (x_j - x_{i_j}) + 1)/2} 
	v_{\boldsymbol \ell,\boldsymbol i} 
	\prod_{j=1}^d
	P_{k_j} \left ( \frac{x_j - x_{i_j} }{h_{\ell_j}} \right )   \\
	& :=
	\sum_{\substack{\| \boldsymbol{\ell} \|_1 \leq n + d - 1 \\ \boldsymbol{i} \in {I}_{\boldsymbol{\ell}}}} 
	\sum_{k_1, \cdots, k_d \in \{0,1\} }
	w_{\boldsymbol \ell,\boldsymbol i} 
	\prod_{j=1}^d
	P_{k_j} \left( \frac{x_j - x_{i_j} }{h_{\ell_j}} \right),
\end{align}
where we have defined
\begin{equation}\label{w-coeffs}
	w_{\boldsymbol \ell,\boldsymbol i} = (-1)^{\sum_{j=1}^d k_j ( \text{sgn} (x_j - x_{i_j}) + 1)/2} 
	v_{\boldsymbol \ell,\boldsymbol i}.
\end{equation}

\begin{remark}
Before discussing the implementation of the LCU subroutine, we address an important point. The sign of the coefficients \(w_{\boldsymbol{\ell}, \boldsymbol{i}}\) in \eqref{w-coeffs} is determined by
\begin{equation}\label{sign-caveat}
    (-1)^{\sum_{j=1}^d k_j (\operatorname{sgn}(x_j - x_{i_j}) + 1)/2}.
\end{equation}
A priori, this expression depends on the input vector \(\boldsymbol{x} = (x_1, \dots, x_d)\). However, this dependence does not pose a significant issue. The partition points define a decomposition of \([0,1]^d\) into hypercubes, with the leftmost corner of each hypercube corresponding to a partition point. Given \(\boldsymbol{x} \in [0,1]^d\), a quantum circuit can be designed to identify the hypercube containing \(\boldsymbol{x}\). If \(\delta > 0\) is an additional precision parameter, the complexity of this circuit is \(\mathcal{O}(\delta \log (1/\delta))\); see \cite[Lemma S14]{yu2024provable} for a detailed argument and tighter bounds. Once the relevant hypercube is identified, all expressions in \eqref{w-coeffs} can be computed efficiently in \(\mathcal{O}(1)\) time.  We do not include this complexity in our results, as it is considered a pre-processing step performed once the input \(\boldsymbol{x} \in [0,1]^d\) is known, prior to executing the quantum circuit. The focus of this paper is the approximation theory of quantum circuits and their complexity in approximating functions, rather than practical pre-processing concerns, although these will be important for real-world implementation.
\end{remark}

We now discuss our main result:

\begin{prop}\label{prop:pqc-implementation-f-sparse}
	Let $d,n \in \mathbb N$ and $p \in \{2,\infty\}$. For each $f \in X^{2,p}([0,1]^d)$, let $f_{n}^{s}(\boldsymbol{x})$ be as in \eqref{align:function-to-be-approximated}.
	There exists a quantum circuit \( U_f(\boldsymbol{x}) \) such that if a measurement corresponding to the Pauli \( \sigma_z \) observable is made only on the first qubit, then
	\begin{equation}
		\bra{0} U_f^\dagger(\boldsymbol x) Z^{(1)} U_f(\boldsymbol x)  \ket{0} = f_n^s(\boldsymbol{x}).
	\end{equation}
	The quantum circuit $U_f(\boldsymbol x)$ has depth $\mathcal{O}(d 2^{n+d} n^{d-1}((n+d) + (d-1)\log_2 n))$ and width $\mathcal{O}(2d + n +(d-1) \log_2 n)$.
\end{prop}

\begin{proof}
	Consider the operator
	\begin{equation}\label{equation:korobov-lcu}
		U_f(\boldsymbol x) = \sum_{\substack{\| \boldsymbol{\ell} \|_1 \leq n + d - 1 \\ \boldsymbol{i} \in {I}_{\boldsymbol{\ell}}}} 
		\sum_{k_1,\cdots,k_d \in \{0,1\}}
		w_{\boldsymbol \ell, \boldsymbol i} 
		\underbrace
		{\bigotimes_{j=1}^d U_{k_j} \left(\frac{x_j - x_{i_j}}{h_{\ell_j}} \right)}_{U_{\boldsymbol i, \boldsymbol{\ell}}(\boldsymbol x)}.
	\end{equation}
	From \cref{prop:sparse-grid-korobov}, the number of terms in the outer summation of \eqref{equation:korobov-lcu} is \( \mathcal{O}(2^n n^{d-1}) \). Hence, the total number of terms is \( M = \mathcal{O}(2^{n+d} n^{d-1}) \). As a result, \cref{pqc-lcu-cheb} implies that the width of $U_f(\boldsymbol{x})$ is
	\begin{equation}
		\mathcal{O}(d+ \log_2 M) = \mathcal{O}(2d + n +(d-1) \log_2 n ).
	\end{equation}
	Furthermore, \cref{lem:qsp-1d-chebyshev} establishes that each unitary operator \( U_{k_j} \) can be implemented using a quantum circuit with width and depth \( \mathcal{O}(1) \), as it corresponds to Chebyshev polynomials of degree at most one. Therefore, the \( \| \cdot \|_1 \)-norm term appearing in the depth of the quantum circuit in \cref{lem:qsp-1d-chebyshev} scales as \( \mathcal{O}(dM) \). As a result, \cref{pqc-lcu-cheb} implies that the depth of $U_f(\boldsymbol{x})$ is
	\begin{equation}
		\mathcal{O}(3dM \log_2 M) = \mathcal{O}(d 2^{n+d} n^{d-1}((n+d) + (d-1)\log_2 n)).
	\end{equation}
	This completes the proof.
\end{proof}

The significance of \cref{prop:pqc-implementation-f-sparse} lies in the provision of an explicit quantum circuit, accompanied by corresponding complexity estimates, for approximating functions in \( X^{2,p}([0,1]^d) \) where \( p \in \{2, \infty\} \). From this, we can now derive the worst-case parameters required to approximate the quantum circuit to arbitrary accuracy.

\begin{prop}\label{prop:approximation-error}
	Let $d \in \mathbb N$ and $\epsilon \in (0,1)$. For each $f \in X^{2,p}([0,1]^d)$ such that $p \in \{2, \infty \}$, there exists a quantum circuit $U_{f,\epsilon}(\boldsymbol x)$ such that if a measurement corresponding to the Pauli \( \sigma_z \) observable is made only on the first qubit, then
	\begin{equation}
		\bra{0} U_{f,\epsilon}^\dagger(\boldsymbol x) Z^{(1)} U_{f,\epsilon} (\boldsymbol x)  \ket{0} = f_n^s(\boldsymbol{x})
	\end{equation}
	such that
	\begin{equation}\label{error-estimate}
		\| f - f_n^s \|_{L^p([0,1]^d)} \leq \epsilon.
	\end{equation}
	The complexity of $U_{f,\epsilon}(\boldsymbol x)$ is characterized as follows:
	\begin{enumerate}
		\item[(i)] The depth of $U_{f,\epsilon} (\boldsymbol x)$ is 
		$\mathcal{O}\left(  
		\frac{d^2 \left(2 \log_2^{3/2}(1/\epsilon)\right)^d}{\epsilon^{1/2} \log_2^{3/2}(1/\epsilon)}
		W \left( \frac{\epsilon^{-1/d}}{d} \log_2^{3/2} \left( \frac{1}{\epsilon} \right) \right)
		\right)$.
		\item[(ii)] The width of $U_{f,\epsilon} (\boldsymbol x)$ is $\mathcal{O} \left( 2d + d W \left( \frac{\epsilon^{- \frac{1}{d}}}{d} \log_2^{\frac{3}{2}} \left( \frac{1}{\epsilon} \right) \right) \right)$.
	\end{enumerate}
\end{prop}

\begin{proof}
	Let $U_{f,\epsilon}(\boldsymbol{x})$ be the unitary operator considered in \cref{prop:pqc-implementation-f-sparse}, and let $f_{n}^{s}(\boldsymbol{x})$ be as in \eqref{align:function-to-be-approximated}. Using the definition of Lambert's $W$ function, we have the following estimate:
	\begin{equation}\label{lamber}
		N = \mathcal{O}(2^n n^{d-1})
		\quad \Longleftrightarrow \quad
		n = \mathcal{O} 
		\left ( 
		\frac{d-1}{\ln 2} W \left( \frac{ \ln 2 N^{\frac{1}{d-1}}}{d-1} \right)
		\right )
		= 
		\Theta 
		\left ( 
		d W \left( \frac{ N^{\frac{1}{d}}}{d} \right)
		\right ).
	\end{equation}
	The last estimate above follows since
	\begin{equation}
		\lim_{d \rightarrow \infty} \frac{(d-1) W\left(\frac{c^{1/(d-1)}}{d-1}\right)}{d W\left(\frac{c^{1/d}}{d}\right)} 
		= \lim_{d \rightarrow \infty} \frac{d-1}{d} \cdot 
		\lim_{d \rightarrow \infty}
		\frac{W\left(\frac{c^{1/(d-1)}}{d-1}\right)}{W\left(\frac{c^{1/d}}{d}\right)}
		=
		1.
	\end{equation}
	for any $c \geq 1$.
	By \cref{prop:sparse-grid-korobov}, if we choose $N$ such that $N = \mathcal{O}(\epsilon^{-\frac{1}{2}} \log^{\frac{3}{2}(d-1)}_2 (1/\epsilon))$, then \eqref{error-estimate} is satisfied. By \eqref{lamber}, we have
	\begin{equation}
		n  = \mathcal{O}\left ( d W \left ( \frac{\epsilon^{- \frac{1}{2d} }}{d} \log^{\frac{3(d-1)}{2d}}_2 \left ( \frac{1}{\epsilon} \right ) \right )  \right )
		= \Theta \left( d W \left ( \frac{\epsilon^{- \frac{1}{d} }}{d} \log^{\frac{3}{2}}_2 \left ( \frac{1}{\epsilon} \right ) \right ) \right).
	\end{equation}
	The last estimate above follows since
	\begin{equation}
		\lim_{d \rightarrow \infty}
		\frac{W \left ( \frac{\epsilon^{- \frac{1}{2d} }}{d} \log^{\frac{3(d-1)}{2d}}_2 \left ( \frac{1}{\epsilon} \right ) \right )}{W \left ( \frac{\epsilon^{- \frac{1}{d} }}{d} \log^{\frac{3}{2}}_2 \left ( \frac{1}{\epsilon} \right ) \right )}=1.
	\end{equation}
	The equation above expresses an upper bound on \( n \) as a function of \( d \). To emphasize this, we write \( n \) as \( n(d) \) in what follows.
	Having expressed \( n \) as a function of \( d \), we now simplify the expression \( n(d) + d \log_2 n(d) \) and argue that \( n(d) + d \log_2 n(d) = \mathcal{O}(n(d)) \). Since \( W(x) \leq x \) for \( x \geq 0 \), we first note that this yields a slightly looser, but simpler upper bound.
	\begin{equation}
		n(d) = \mathcal{O}\left ( \frac{d}{d} \epsilon^{- \frac{1}{d} } \log^{\frac{3}{2}}_2 \left ( \frac{1}{\epsilon} \right )  \right ) = \Theta ( \epsilon^{- \frac{1}{d} } ).
	\end{equation}
	The final estimate follows since \( \log_2^{3/2}(1/\epsilon) = \mathcal{O}(1) \) for fixed \( \epsilon \in (0,1) \). The claim that $n(d) + d \log_2 n(d) = \mathcal{O}(n(d))$ now follows from the following computation:
	\begin{equation}
		\lim_{d \to \infty} \frac{\epsilon^{- \frac{1}{d} } + d \log_2( \epsilon^{- \frac{1}{d} }) }{\epsilon^{- \frac{1}{d} }} = 
		1 + \lim_{d \to \infty} \frac{d \log_2( \epsilon^{- \frac{1}{d} }) }{\epsilon^{- \frac{1}{d} }} = 1 + \frac{\ln (1/\epsilon) }{ \ln 2} = \Theta(1).
	\end{equation}\label{depth-estimate}
	By \cref{prop:pqc-implementation-f-sparse}, the depth of $U_{f,\epsilon}(\boldsymbol x)$ is 
	\begin{align}
		\mathcal{O}\left( d 2^{n+d} n^{d-1}(n + d \log_2 n) \right) 
		&= \mathcal{O}\left( 2^d N dn(d) \right) \\
		&= \Theta \left(  
		\frac{d^2 \left(2 \log_2^{3/2}(1/\epsilon)\right)^d}{\epsilon^{1/2} \log_2^{3/2}(1/\epsilon)}
		W \left( \frac{\epsilon^{-1/d}}{d} \log_2^{3/2} \left( \frac{1}{\epsilon} \right) \right)
		\right).
	\end{align}
	Here we have used the estimates that
	\begin{align}
		2^d N  & =  \mathcal{O}(2^d \epsilon^{-\frac{1}{2}} \log^{\frac{3}{2}(d-1)}_2 (1/\epsilon)) =  \Theta \left( \frac{(2 \log^{\frac{3}{2}}_2(1/\epsilon))^d}{\epsilon^{\frac{1}{2}} \log^{\frac{3}{2}}_2 (1/\epsilon))}  \right) , \\
		dn(d) & =
		\mathcal{O} \left( d^2 W \left ( \frac{\epsilon^{- \frac{1}{d} }}{d} \log^{\frac{3}{2}}_2 \left ( \frac{1}{\epsilon} \right ) \right )\right).
	\end{align}
	Since the following limit can be verified by a straightforward computation,
	$$
	\lim_{d \to \infty} \frac{2d+n(d)+(d-1)\log_2n(d)}{2d+n(d)} = 1,
	$$
	the width of $U_{f,\epsilon}(\boldsymbol x)$ is $\mathcal{O}(2d + n(d) + (d-1) \log_2 n(d)) = \Theta(2d + n(d)) $. In particular, the width is
	\begin{align}
		\mathcal{O} \left( 2d + d W \left( \frac{\epsilon^{- \frac{1}{d}}}{d} \log_2^{\frac{3}{2}} \left( \frac{1}{\epsilon} \right) \right) \right).
	\end{align}
	This completes the proof.
\end{proof}


We now treat the case $p \notin \{2, \infty \}$ separately. This is because the optimal sparse grid decomposition (discussed in the previous section) is derived by solving the corresponding optimization problem with respect to the $L^p([0,1]^d)$ norms for $p \in \{2 , \infty\}$. 
Therefore, the sparse grid decomposition in \eqref{equation:sparse-grid} may not be the optimal decomposition when $p \notin \{2, \infty \}$. 
Since any function in $X^{2,p}([0,1]^d)$ when $p \notin \{2, \infty \}$ can still be written as 
\begin{equation}
	f(\boldsymbol x) = \sum_{\boldsymbol \ell} \sum_{\boldsymbol i \in {\boldsymbol I}_{\boldsymbol \ell}} v_{\boldsymbol \ell,\boldsymbol i} \phi_{\boldsymbol \ell,\boldsymbol i}(\boldsymbol x),
\end{equation}
we can still in practice use the decomposition discussed above to approximate any $f \in X^{2,p}([0,1]^d)$ when $p \notin \{2, \infty \}$ by $f_{n}^{s}(\boldsymbol{x})$.

\begin{cor}\label{cor:compplexity-estimate-2}
	Let $d \in \mathbb N$ and $\epsilon \in (0,1)$. For each $f \in X^{2,p}([0,1]^d)$ such that $2 < p < \infty$, there a exists a quantum circuit $U_{f,\epsilon} (\boldsymbol x)$ 
	(as in \cref{prop:approximation-error})
	such that if a measurement corresponding to the Pauli \( \sigma_z \) observable is made only on the first qubit, then
	\begin{equation}
		\bra{0} U_{f,\epsilon}^\dagger(\boldsymbol x) Z^{(1)} U_{f,\epsilon}(\boldsymbol x)  \ket{0} = f_n^s(\boldsymbol{x})
	\end{equation}
	such that we have
	\begin{equation}\label{error-est-2}
		\| f - f_n^s \|_{L^p([0,1]^d)} \leq \epsilon.
	\end{equation}
	The complexity of $U_{f,\epsilon}(\boldsymbol x)$ is characterized as follows:
	\begin{enumerate}
		\item[(i)] The depth of $U_{f,\epsilon} (\boldsymbol x)$ is \begin{equation}
			\mathcal{O} \left(  
			d^2 (12\beta \log_2 \beta)^{\beta} 
			\alpha^{\beta} \epsilon^{-\frac{p}{2p-1}}
			\log_2^{\beta} \left(\frac{1}{\epsilon}\right)
			W\left(\frac{(6\beta \log_2 \beta)^{\alpha} 
				\alpha^{\alpha} \epsilon^{-\frac{p}{d(2p-1)}}
				\log_2^{\alpha} \left(\frac{1}{\epsilon}\right)}{d}\right)
			\right).
		\end{equation}    
		\item[(ii)] The width of $U_{f,\epsilon} (\boldsymbol x)$ is 
		\begin{equation}
			\mathcal{O} \left( 2d +  dW\left(\frac{(6\beta \log_2 \beta)^{\alpha}  \alpha^{\alpha} \epsilon^{-\frac{p}{d(2p-1)}} \log_2^{\alpha} \left(\frac{1}{\epsilon}\right)}{d}\right) \right),
		\end{equation}
	\end{enumerate}
	where we have defined $\alpha=(3p-1)/(2p-1)$ and $\beta = \alpha (d-1)$. 
\end{cor}

\begin{proof}
	A computation in \cite[Equation 6.3]{mao2022approximation} shows that
	\begin{equation}
		\| f - f^s_n \|_{L^p([0,1]^d)} 
		\leq (\log_2 N)^{\left( 3 - \frac{1}{p} \right) (d-1)} N^{-\left(2 - \frac{1}{p}\right)}.
	\end{equation}
	Following the derivation presented in \cite[Page 8, below (3.2)]{fang2025korobovCNN}, to guarantee an accuracy $\epsilon$ in \eqref{error-est-2} we need to choose $N$ such that
	\begin{equation}
		N = \left\lceil 
		(6\beta \log_2 \beta)^{\beta} 
		\alpha^{\beta} \epsilon^{-\frac{p}{2p-1}}
		\log_2^{\beta} \left(\frac{1}{\epsilon}\right)
		\right\rceil,
	\end{equation}
	where
	$\alpha=(3p-1)/(2p-1)$ and $\beta = \alpha (d-1)$. 
	The unitary operator is constructed as in \cref{prop:approximation-error}. We now derive the complexity of the quantum circuit in this case. Using the definition of Lambert's $W$ function and \eqref{lamber}, it suffices to choose 
	\begin{align}
		n=\mathcal{O}\left(dW\left(\frac{N^{\frac{1}{d-1}}}{d}\right)\right)&=
		\mathcal{O}\left(dW\left(\frac{(6\beta \log_2 \beta)^{\alpha} 
			\alpha^{\alpha} \epsilon^{-\frac{p}{(d-1)(2p-1)}}
			\log_2^{\alpha} \left(\frac{1}{\epsilon}\right)}{d}\right)\right), \\
		&=\Theta\left(dW\left(\frac{(6\beta \log_2 \beta)^{\alpha} 
			\alpha^{\alpha} \epsilon^{-\frac{p}{d(2p-1)}}
			\log_2^{\alpha} \left(\frac{1}{\epsilon}\right)}{d}\right)\right).
	\end{align}
	The equation above provides an upper bound on \( n \) as a function of \( d \). To highlight this relationship, we denote it as \( n(d) \) as before.
	Once again, we have \( n(d) + d \log_2 n(d) = \mathcal{O}(n(d)) \). To see this, note that since \( W(x) \leq x \) for \( x \geq 0 \), we obtain a slightly looser bound on \( n \) given by
	\begin{equation}
		n=
		\mathcal{O}\left(\frac{d}{d}(6\beta \log_2 \beta)^{\alpha} 
		\alpha^{\alpha} \epsilon^{-\frac{p}{d(2p-1)}}
		\log_2^{\alpha} \left(\frac{1}{\epsilon}\right)\right)
		=\Theta(( d \log_2 (\alpha d))^{\alpha} \epsilon^{-\frac{p}{d(2p-1)}}).
	\end{equation}
	The final estimate follows since $\alpha^\alpha=\mathcal{O}(1)$ and $\log^\alpha_2(1/\epsilon)=\mathcal{O}(1)$ for fixed $\epsilon \in (0,1)$ and $2 \leq p \leq \infty$. The claim that $n(d) + d \log_2 n(d) = \mathcal{O}(n(d))$ now follows from the following computation:
	\begin{equation}
		\lim_{d \to \infty} \frac{n(d) + d \log_2 n(d)}{n(d)}
		= 
		1 + \lim_{d \to \infty} \frac{d \log_2( ( d \log_2 (\alpha d))^{\alpha} \epsilon^{-\frac{p}{d(2p-1)}}) }{( d \log_2 (\alpha d))^{\alpha} \epsilon^{-\frac{p}{d(2p-1)}}} = \Theta(1).
	\end{equation}
	Since the depth of \( U_{f,\epsilon}(\boldsymbol{x}) \) is \( \mathcal{O}\left( 2^d N d n(d) \right) \), an upper bound on the depth of \( U_{f,\epsilon}(\boldsymbol{x}) \) is therefore given by
	\begin{align}
		\Theta \left(  
		2^d d^2 (6\beta \log_2 \beta)^{\beta} 
		\alpha^{\beta} \epsilon^{-\frac{p}{2p-1}}
		\log_2^{\beta} \left(\frac{1}{\epsilon}\right)
		W\left(\frac{(6\beta \log_2 \beta)^{\alpha} 
			\alpha^{\alpha} \epsilon^{-\frac{p}{d(2p-1)}}
			\log_2^{\alpha} \left(\frac{1}{\epsilon}\right)}{d}\right)
		\right).
	\end{align}
	Since \( 2^d = \mathcal{O}(2^{\alpha(d-1)}) = \mathcal{O}(2^{\beta})  \), we can absorb the factor \( 2^d \) into \( (6\beta \log_2 \beta)^{\beta} \) by replacing \( 2^d \)  by \( 2^\beta \), and instead use the following estimate:
	\begin{align}
		\Theta \left(  
		d^2 (12\beta \log_2 \beta)^{\beta} 
		\alpha^{\beta} \epsilon^{-\frac{p}{2p-1}}
		\log_2^{\beta} \left(\frac{1}{\epsilon}\right)
		W\left(\frac{(6\beta \log_2 \beta)^{\alpha} 
			\alpha^{\alpha} \epsilon^{-\frac{p}{d(2p-1)}}
			\log_2^{\alpha} \left(\frac{1}{\epsilon}\right)}{d}\right)
		\right).
	\end{align}
	Since the width of $U_{f,\varepsilon}$ is $\mathcal{O}(2d + n(d) + (d-1) \log_2 n(d)) = \Theta(2d + n(d))$, an upper bound on the width of \( U_{f,\epsilon}(\boldsymbol{x}) \) is therefore given by
	\begin{equation}
		\mathcal{O}
		\left(
		2d + 
		dW\left(\frac{(6\beta \log_2 \beta)^{\alpha} 
			\alpha^{\alpha} \epsilon^{-\frac{p}{d(2p-1)}}
			\log_2^{\alpha} \left(\frac{1}{\epsilon}\right)}{d}\right)
		\right).
	\end{equation}
	This completes the proof.
\end{proof}




\section{Conclusion}\label{future}
By leveraging the QSP and LCU algorithms, we have constructed quantum circuits capable of approximating functions in the Korobov function space. This work establishes a rigorous theoretical foundation for the quantum implementation of a broad class of multivariate functions. The main contributions include precise estimates of the approximation error and a detailed complexity analysis of the proposed quantum circuits. These results advance the emerging field of \emph{quantum neural network approximation theory} for parameterized quantum circuits and contribute to research at the interface of quantum computing and scientific computing.

We emphasize that the central aim of \emph{quantum neural network approximation theory} is to investigate the fundamental capabilities and limitations of (parameterized) quantum circuits in approximating functions from various function spaces, independent of the specifics of training algorithms or data availability. This theoretical framework captures only a subset of the broader practical challenge. A complete analysis of the utility of (parameterized) quantum circuits in real applications must also consider additional factors, such as quantum state initialization, the quantity and quality of training data in learning-based frameworks, the optimization of circuit parameters, and the fidelity of measurement and readout. We anticipate that future research will address these aspects to deepen understanding of quantum-enhanced function approximation.

Several concrete future directions include:

\begin{enumerate}
    \item[(i)] \textbf{Different Function Spaces:} Investigate whether (parameterized) quantum circuits can approximate functions in other function spaces, such as Besov spaces \cite{liu2021besov}, Bochner spaces \cite{abdeljawad2022approximations}, and Fr\'echet spaces \cite{benth2023neural}.
    
    \item[(ii)] \textbf{Optimal Sparse Grid Decomposition:} For $p \notin \{2, \infty\}$, the optimal sparse-grid decomposition for the $L^p$ norm has not been computed. Determining it may yield better complexity estimates than those in \cref{cor:compplexity-estimate-2}.
    
    \item[(iii)] \textbf{Other Basis Functions:} The sparse grid decomposition is not limited to the hierarchical basis used here. Extensions to other multiscale bases, such as wavelets, could allow quantum circuit approximations of functions in diverse spaces.
    
    \item[(iv)] \textbf{Higher-Order Korobov Spaces:} Our analysis focuses on \(X^{r,p}([0,1]^d)\) with \(r=2\). Extensions to \(r \ge 3\) are feasible. \cite[Section 4.2]{bungartz2004sparse} uses higher-degree polynomial bases, including Lagrange polynomials, for \(r \ge 4\). \cite{li2025higherorderapproximationrates} applies this to two-dimensional CNNs. Lagrange polynomials can be expanded in Chebyshev polynomials when using Chebyshev nodes, so results in \cref{cheb-qc} are applicable, with the caveat that a non-uniform sparse grid is required. Error analysis for this method remains an open problem.
    
    \item[(v)] \textbf{State Preparation:} We used a simple argument to construct a unitary for state preparation (\cref{equation:state-preperation}). Leveraging results from quantum state preparation literature may improve the associated complexity estimate.
    
    \item[(vi)] \textbf{Implementing LCU:} The primary overhead in our complexity estimate arises from the controlled unitary in the LCU algorithm. Exploring quantum circuit depth minimization techniques could reduce this cost.
    
    \item[(vii)] \textbf{Variants of LCU:} Several LCU variants have been proposed. For example, \cite{chakraborty2024implementing} introduces methods suitable for intermediate-term quantum computers. Investigating such variants could enable new quantum circuits for approximating $d$-dimensional functions across various function spaces.
\end{enumerate}


\printbibliography
\end{document}